\documentclass[a4paper,12pt]{article}
\usepackage{latexsym,amssymb,amsfonts,amsmath,amsthm}
\usepackage[dvips]{graphicx}
\usepackage{color}

\renewcommand{\theequation}{\thesection.\arabic{equation}}

\newtheorem{theorem}{Theorem}
\newtheorem{lemma}{Lemma}

\newtheorem{remark}{Remark}
\newtheorem{definition}{Definition}
\newtheorem{conjecture}{Conjecture}

\numberwithin{theorem}{section}
\numberwithin{lemma}{section}
\numberwithin{corollary}{section}
\numberwithin{proposition}{section}
\numberwithin{remark}{section}

\newcommand{\bs}[1]{\boldsymbol{#1}}
\newcommand{\im}{\bs{\rm i}}

\newcommand{\segawa}{\textcolor{black}}
\newcommand{\red}{\textcolor{black}}

\title{{\Large {\bf Phase measurement of quantum walks: application to structure theorem of the positive support of the Grover walk}}
\author{
{\small Norio KONNO}\\
{\scriptsize Department of Applied Mathematics, 
Faculty of Engineering, Yokohama National University}\\
{\scriptsize Hodogaya, Yokohama 240-8501, Japan}\\
{\small Iwao SATO} \\ 
{\scriptsize Oyama National College of Technology} \\ 
{\scriptsize Oyama, Tochigi 323-0806, Japan} \\
{\small Etsuo SEGAWA\footnote{Email: e-segawa@m.tohoku.ac.jp}} \\
{\scriptsize Graduate School of Information Sciences, Tohoku University} \\
{\scriptsize Sendai 980-8579, Japan}.}
}
\date{\empty }
\begin{document}
\maketitle

\par\noindent
\begin{small}
\par\noindent
{\bf Abstract}. 
We obtain a structure theorem of the positive support of the $n$-th power of the Grover walk on $k$-regular graph whose girth is greater than $2(n-1)$. 
This structure theorem is provided by the parity of the amplitude of another quantum walk on the line which depends only on $k$. 
The phase pattern of this quantum walk has a curious regularity. 
We also exactly show how the spectrum of the $n$-th power of the Grover walk is obtained by lifting up that of the adjacency matrix to the 
complex plain.  
\end{small}

\section{Introduction}
The Grover walk is one of the important quantum walk model for not only quantum search algorithm~\cite{Ambainis2003,Kempe2003,Kendon2007} 
but also bridges connecting to the quantum graphs~\cite{SchanzSmilansky,Tanner}, reversible random walks~\cite{HKSS_JFA}, 
and also graph zeta~\cite{KS2011,RenETAL} and graph theory~\cite{Portugal}. 
The time evolution of the Grover walk is explained by a discrete-time analogue of 
reflection and transmission of the wave at each junction, that is, the vertex~\cite{FH}. 
To give more precise definition, we prepare notions of graphs. 
For given $G=(V,E)$, let $A=A(G)$ be the set of symmetric arcs induced by the edge set $E$. 
The inverse arc of $a\in A$ is denoted by $\bar{a}$ and the terminal and original vertices of $a$ are denoted by $t(a)$ and $o(a)$, respectively. 
For any $a\in A$, $|a|$ is the support edge of $a$, thus $|\bar{a}|=|a|\in E$. 
The degree of $u\in V$ is defined by $\mathrm{deg}(u)=|\{a\in A \;|\; t(a)=u\}|$. 
The total Hilbert space $\mathcal{A}$ is generated by the symmetric arc set $A$ of the given graph $G$. 
The quantum coin assigned at each vertex $u$ produces the complex valued weight of the transmission and reflection rate so that 
this representation matrix is a $\mathrm{deg}(u)$-dimensional unitary operator. 
In particular, for the Grover walk case, the transmission rate is $2/\mathrm{deg}(u)$, and the reflection rate is 
$2/\mathrm{deg}(u)-1$. 
Then the Grover walk is defined as follows: 
\begin{definition}Grover walk on $G=(V,A)$
\begin{enumerate}
\item The total Hilbert space: $\mathcal{A}:=\ell^2(A)=\{\psi: A\to \mathbb{C} \;|\; ||\psi||<\infty \}$.
Here the inner product is the standard inner product, that is, $\langle \psi,\varphi \rangle=\sum_{a\in A}\overline{\psi(a)}\varphi(a)$. 
\item Time evolution $U: \mathcal{A}\to \mathcal{A}$ (unitary)
	\[ (U\psi)(a)=\sum_{b:t(b)=o(a)} \left( \frac{2}{\deg(o(a))}-\delta_{b,\bar{a}} \right)\psi(b). \]
\end{enumerate}
\end{definition}

Let the the total time evolution operator of the Grover walk and the $n$-th iteration of the Grover walk starting from the initial state $\psi_0$ 
be denoted by $U$ and $\psi_n$, respectively.  
We introduce two non-linear maps $\mu: \mathcal{A}\to \ell(A)$ and $\nu: \mathcal{A}\to [0,2\pi]^A$  as follows: 
for $\psi\in \mathcal{A}$, 
	\begin{align}
        (\mu(\psi))(a) &= |\psi(a)|^2, \\
        (\nu(\psi))(a) &= \arg (\psi(a)). 
        \end{align}
Due to the unitarity of the Grover walk, $\mu_n:= \mu(\psi_n)$ becomes a probability distribution when the norm of the initial state $\psi_0$ is unit. 
Main interest of the Grover walk has been the investigation of the sequence of $\mu_n$'s: 
the typical behaviors of quantum walks drive from observing the behavior $\mu_n$, for example, 
the coexistence of linear spreading and localization e.g., \cite{Konno2008b,Suzuki} 
and its stationary measure for infinite graphs e.g., \cite{KonnoTakei}, 
the efficiency to the quantum search algorithm e.g., \cite{Portugalbook} and its references therein
and perfect state transfer e.g., \cite{Godsil,Stefanak}. 
However it seems to be natural to investigate also the phase measurement $\nu_n$'s. 
Indeed, we focus on this $\nu$ in this paper since the map $\nu$ plays a key role 
to give the structure theorem of the positive support of the Grover walk. 
Here for the real matrix $M$, the positive support of $M$; $S(M)$, is defined by
	\begin{equation}
        (S(M))(b,a) = \begin{cases} 1 & \text{: $(M)_{b,a}>0$} \\ 0 & \text{: $(M)_{b,a}\leq 0$} \end{cases}
        \end{equation} 
Taking the positive support of the Grover walk is first motivated by the fact of a direct connection  
between the Ihara zeta function and the positive support of the time evolution of the Grover walk~\cite{RenETAL}: 
	\[ \zeta_G(u)=\det(1-uS(U_G))^{-1}, \]
where $\zeta_G(u)$ is the Ihara zeta function and $U_G$ is the time evolution operator of the Grover walk induced by graph $G$. 
Zeta function of a graph was started from Ihara zeta function of a graph \cite{Ihara1966}. 
Originally, Ihara defined $p$-adic Selberg zeta functions of 
discrete groups, and showed that its reciprocal is a explicit polynomial. 
Serre \cite{Serre} pointed out that the Ihara zeta function is the zeta function of 
the quotient $T/ \Gamma $ (a finite regular graph) of the one-dimensional 
Bruhat-Tits building $T$ (an infinite regular tree) associated with $GL(2, k_p)$. 
A zeta function of a regular graph $G$ associated with a unitary 
representation of the fundamental group of $G$ was developed by 
Sunada \cite{Sunada1986,Sunada1988}. 
Hashimoto \cite{Hashimoto1989} treated multivariable zeta functions of bipartite graphs. 
Bass \cite{Bass1992} generalized Ihara's result on the zeta function of 
a regular graph to an irregular graph, and showed that its reciprocal is 
again a polynomial. 
New proofs for Bass' formula were given in \cite{FZ1999,KS2000,ST1996}. 
Furthermore, Konno and Sato \cite{KS2011} presented a explicit formula for the characteristic 
polynomial of the Grover walk on a graph $G$ by using the determinant expression for 
the second weighted zeta function of $G$, and directly obtained spectra for the Grover walk on $G$. 

The second motivation to take the positive support to the Grover walk operator is that 
the spectrum of the positive support of the Grover walk has been believed to be a strong tool for the graph isomorphism problem: 
\begin{conjecture}\cite{EmmsETAL2006}
Let $G$ and $H$ be strongly regular graphs. Then
\[ G\cong H \Leftrightarrow \sigma(S(U_G^3))=\sigma(S(U_H^3)). \]
\end{conjecture} 
\noindent Recently, a counter example of the graph having a large number of vertices 
are suggested by a combination of theoretical and numerical method~\cite{GGM}. 
However finding the class of strongly regular graphs which conserves the conjecture is still an interesting open problem. 

Therefore taking together with the above two motivations of the positive support of the Grover walk, 
we can naturally extend the Ihara zeta function as follows: 
	\[ \zeta^{(n)}_G(u)=\det(1-uS(U^n))^{-1}. \]
There are structure theorems for $S(U^2)$ and $S(U^3)$ as follows: 
\begin{theorem}Structure theorem for $S(U^2)$ and $S(U^3)$
\begin{enumerate}
\item (\cite{GG})
If $G$ is a graph without leaves, then 
	\[ S(U^2)=I+S(U)^2; \]
\item (\cite{HKSS}) additionally, if the girth is greater than $4$; $g(G)>4$, and it is $k$-regular, then  
	\[ S(U^3)=S(U)^3+{}^TS(U). \]
Here the girth $G$ is the smallest length of cycle of $G$.
\end{enumerate}
\end{theorem}
\noindent Moreover a beautiful structure theorem of $S(U^3)$ in~\cite{Guo} for the strongly regular graph is obtained. 
In this paper, we consider the structure theorem of $S(U^n)$ for general $n$. 
The graph in our setting should have a large girth $g(G)>2(n-1)$ with the degree's regularity which includes the setting of \cite{HKSS}. 

In our main theorem, $S(U^n)$ is expressed by a linear combination of $S(U)^k$, $JS(U)^k$, ${}^TS(U)^k$ and $J\;{}^TS(U)^k$ $(k=0,\dots,n)$, 
where $J:\ell^2(A)\to \ell^2(A)$ such that $(J\psi)(a)=\psi(\bar{a})$, 
${}^TM$ is the transpose of $M$. 
The advantage point of this expression is that we can expressed the spectrum of $S(U^n)$ 
by using the spectrum of the adjacency matrix of $G$. 
Thus we can see how the spectrum of $S(U^n)$ is lifted up to the complex plane from the spectrum of the adjacency matrix on the real line. 
See Theorem~\ref{thm:eigenorbit} and Fig.~\ref{Fig.4}, we obtain the support of the non-trivial zero's of $1/(z^{|A|}\zeta^{(n)}_G(z^{-1}))$. 
On the other hand, the negative point is that $\sigma(S(U^n))$ with $g(G)>2(n-1)$ cannot determine the graph isomorphism 
since there are graphs which are not isomorphism but cospectral of the positive support due to this ``advantage" point. 
As is suggested by the appearance of the Hadamard product in \cite{Guo}, only the spectrum of the adjacency matrix cannot determine $S(U^n)$ in general.
However we believe that our main theorem brings a new study motivation of quantum walks investigating the phase observation $\nu$;
note that the operation taking support is converted to the phase observation problem, that is, 
\segawa{letting $\psi_n(a)$ be the $n$-th iteration of the Grover walk at $a$ starting from $b$, that is, $\psi_n(a)=(U^n\delta_b)(a)$, 
then we have $S(U^n))_{a,b}=1 \Leftrightarrow \nu(\psi_n)(a)=0$.
}

We show that in our setting, this phase observation problem can be switched to solving 
the phase pattern $\{\nu_n\}_{n\in \mathbb{N}}$ of the discriminant quantum walk on the one-dimensional lattice defined below, 
which is another quantum walk model. 
As we will see, this phase pattern informs us an exact expression for the structure theorem
and seems to have a curious regularity (see Figs.~\ref{Fig.2}--\ref{Fig.3}). 
However a complete decode of this pattern still remains as one of the interesting open problems induced by the operation of the positive support. 
The following is the definition of the discriminant quantum walk: 
\begin{definition}\label{DesQW}Discriminant quantum walk. Let $k\geq 2$ be a natural number. 
\begin{enumerate}
\item Hilbert space: $\ell^2(\mathbb{Z};\mathbb{R}^2)$
\item Time evolution: 
We assign the coin operator at each $x\in \mathbb{Z}$ depending on the positive or negative side. 
\[ H_m(x)=\begin{cases}
	\begin{bmatrix}   2\sqrt{k-1}/k & -1+2/k \\ 1-2/k  & 2\sqrt{k-1}/k  \end{bmatrix} & \text{: $x\geq 0$} \\
        \\
        \begin{bmatrix}   2\sqrt{k-1}/k & 1-2/k \\ -1+2/k  & 2\sqrt{k-1}/k  \end{bmatrix} & \text{: $x<0$}
        \end{cases}
\]
The time evolution of the discriminant quantum walk $W$ is described as $\phi_{n+1}=W\phi_n$ 
with the initial state $\phi_0(x)=\delta_0(x)|R\rangle$ such that 
	\[ (W\phi)(x)=P(x+1)\phi(x+1)+Q(x-1)\phi(x-1), \]
where $P(x)=|L\rangle\langle L| H_m(x)$, $Q(x)=|R\rangle\langle R| H_m(x)$.
Here $|L\rangle={}^T[1,0]$, $|R\rangle={}^T[0,1]$ and $\langle R|={}^*|L\rangle$, $\langle R|={}^*|R\rangle$. 
\item Phase measure: Letting $\phi\in \ell^2(\mathbb{Z};\mathbb{R}^2)$, we put $\phi(x):= {}^T[\phi(x;L),\phi(x;R)]\in \mathbb{R}^2$. 
Then the phase measure $\nu:\ell^2(\mathbb{Z};\mathbb{R}^2)\to \{0,\pi,\emptyset\}^{\mathbb{Z}\times \{L,R\}}$ is the operation to pick up the 
argument of $\phi\in \ell^2(\mathbb{Z};\mathbb{R}^2)$, that is,  
	\[ (\nu(\phi))(x;N) = \begin{cases}  0 & \text{: $\phi(x;N)>0$,}\\ \emptyset & \text{: $\phi(x;N)=0$,}\\ \pi &  \text{: $\phi(x;N)< 0$.} \end{cases}\;\;(N\in \{L,R\})\]
\end{enumerate}
\end{definition}
Now we are ready to give our main theorem as follows: 
\begin{theorem}\label{mainThm}
Let $\phi_n\in \ell^2(\mathbb{Z};\mathbb{C}^2)$ and $\nu:\ell^2(\mathbb{Z};\mathbb{C}^2)\to \{0,\pi,\emptyset\}^{\mathbb{Z}\times \{L,R\}}$ be the above. 
Under the assumption $g(G)>2(n-1)$, and the regularity $k$, we have 
	\begin{equation}\label{mainThmEq1}
        S(U^n)=\sum_{j=0}^n \left(\epsilon_j S(U)^j +\tau_j JS(U)^j\right)
        	+ \sum_{j=1}^{n-1} \left(\epsilon_{-j} {}^TS(U)^j +\tau_{-j} J\;{}^T(S(U)^j)\right).
        \end{equation}
Here $\epsilon_j, \tau_j\in \{0,1\}$ is determined by 
	\begin{align}
        \epsilon_j=
        \begin{cases}
        1 & \text{: $(\nu(\phi_n))(j;R)=0$} \\
        0 & \text{: otherwise}
        \end{cases}, \;
        \tau_j=
        \begin{cases}
        1 & \text{: $(\nu(\phi_n))(j-1;L)=0$} \\
  	0 & \text{: otherwise.}
        \end{cases}
        \end{align}
\end{theorem}
This paper is organised as follows. 
We provide the phase pattern of the discriminant quantum walk in section~2, which is useful to obtain the exact form of RHS of Theorem~\ref{mainThm} 
for each $n$. We put $\nu_n:=\nu(\phi_n)$. The sequence of $\{\nu_n\}_n$ seems to depict a kind of interesting regular pattern when we see it up to large time step. 
Section~3 is devoted to show the proof of our main theorem. 
Finally we give the spectral orbit of $S(U^n)$ with respect to the adjacency matrix of $G$ for $g(G)>2(n-1)$ and the degree $k$. 
\section{Demonstration of the phase pattern}
In this section, we provide the phase pattern of the discriminant quantum walk. 
By Theorem~\ref{mainThm}, it is convenient to give the following one-to-one correspondence between 
$\{(j;R),(j+L) \;|\; j\in \mathbb{Z}\}$ and $\{ S(U)^j, {}^T(S(U)^j), JS(U)^j, J\;{}^T(S(U)^j) \;|\; j\in \mathbb{Z}_{\geq 0} \}$
\begin{align} 
(j;R) & \leftrightarrow \begin{cases} S(U)^j & \text{: $j\geq 0$}\\ {}^T(S(U)^{|j|}) & \text{: $j<0$} \end{cases} \notag\\
(j;L) & \leftrightarrow \begin{cases} JS(U)^{j+1} & \text{: $j\geq 0$}\\ J\;{}^T(S(U)^{|j+1|}) & \text{: $j<0$} \end{cases}\label{oneone}
\end{align}
For $n=1$, computing $\phi_1=W\phi_0$, we have $\phi_1(1)=Q_+|R\rangle$ and $\phi_1(-1)=P_+|R\rangle$, that is, 
	\[ \phi_1(j)=\begin{cases}  {}^T[ 0\;\; 2\sqrt{k-1}/k ] & \text{: $j=1$;}\\ 
                                    {}^T[ -1+2/k\;\; 0] & \text{: $j=-1$;} \\ 
                                    0  & \text{: otherwise,} 
        	     \end{cases}\]
where we put $Q_+:=Q(x)$, $P_+:=P(x)$ for $x\geq 0$ and $Q_-:=Q(x)$, $P_-:=P(x)$ for $x<0$. 
Thus the phase measurement results are 
	\begin{align*} 
        (\nu(\phi_1))(j;R) = \begin{cases} 0 & \text{: $j=1$} \\ \emptyset & \text{: otherwise} \end{cases},\;\;
        (\nu(\phi_1))(j;L) = \begin{cases} \pi & \text{: $j=-1$} \\ \emptyset & \text{: otherwise} \end{cases}
        \end{align*}
Then using the relation (\ref{oneone}), since $(1;R)$ is the unique arc such that the phase value is $0$, then we have the 
trivial equation $S(U)=S(U)$. 
In the same way, for $n=2$, since $\phi_2(-2)=P_{-}P_{+}|R\rangle$, $\phi_2(2)=Q_{+}^2|R\rangle$ and $\phi_2(0)=(Q_{-}P_{+} + P_{+}Q_{+})|R\rangle$, 
examining the phases and using the relation (\ref{oneone}), we have 
	\[ S(U^2)=I+S(U)^2. \]
We put $\nu_n:=\nu(\phi_n)$. Now our interest is the sequence of $\{\nu_n\}_{n\in \mathbb{N}}$. 
Figure~\ref{Fig.1} depicts the $\nu_n$'s up to $n=4$, that is, $\{\nu_1,\nu_2,\nu_3,\nu_4\}$. 
Each cell corresponds to $(n,\; (j;N))$ with $n\in \mathbb{N}$, $j\in \mathbb{Z}$ and $N\in \{L,R\}$. 
If $\nu_n(j;N)=0$, that is, $\epsilon_j=1$ for $N=R$ and $\tau_j=1$ for $N=L$, then the corresponding cell color is black, otherwise the color is white. 
Referring the pattern of Fig.~1, we can easily obtain the structure theorem for $n=3$ and $n=4$: 
	\begin{align}
        S(U^3) &= {}^TS(U)+S(U)^3, \label{SU3}\\
        S(U^4) &= {}^TS(U)^2 + I + S(U)^4. \label{SU4}
        \end{align}
After $n=5$, a condition analysis arises with respect to the degree $k$ for example, 
	\begin{align}
        S(U^5) &=
  \begin{cases}
  {}^TS(U)^3 +{}^TS(U)+S(U)+S(U)^5 & \text{: $3\leq k\leq 6$,} \\
  {}^TS(U)^3 +J\;{}^TS(U)^2+{}^TS(U)+S(U)+JS(U)^2+S(U)^5 & \text{: $k \geq 7$.}
  \end{cases} \label{SU5}\\
  \\
  	S(U^6) &=
  \begin{cases}
  {}^TS(U)^4+{}^TS(U)^2+I+S(U)^2+S(U)^6 & \text{: $k=3,4$,}\\
  {}^TS(U)^4+J\;{}^TS(U)^3+{}^TS(U)^2+I+S(U)^2+JS(U)^3+S(U)^6 & \text{: $5\leq k\leq 11$,}\\
  {}^TS(U)^4+J\;{}^TS(U)^3+I+S(U)^2+JS(U)^3+S(U)^6 & \text{: $12\geq k$.}\\
  \end{cases}\label{SU6}
        \end{align}
Figures~\ref{Fig.2}(a)--\ref{Fig.2}(c) are the phase patterns of $\nu_n$'s for $k=20$ up to $n=10,20$ and $100$, respectively. 
According to the result on the phase pattern, 
we can divide $\Xi:=\{ (n,x)\in \mathbb{N}\times \mathbb{Z} \;|\; |x| \leq 2n \}$ plane into three regions $(A)$, $(B)$ and $(C)$: 
there exists $0<c<1$ such that 
\begin{enumerate}
\item Region $(A)$: around the origin;
\item Region $(B)$: $\{ (n,x)\in \Xi \;|\; 0 < |x|/n \leq 2c \}$;
\item Region $(C)$: $\{ (n,x)\in \Xi \;|\; 2c < |x|/n \leq 2 \}$. 
\end{enumerate}
In Region $(A)$, we can observe a kind of check pattern, and in Region $(C)$, there is some regularity while complex pattern appears in Region $(B)$. 
The value $c$ seems to be $2\sqrt{k-1}/k$ which is the diagonal part of the discriminant quantum walk's coin. 
See Figures~\ref{Fig.3}(a)--\ref{Fig.3}(c) for $k=20$, $k=10$ and $k=3$ cases until $n=500$. 
We can show the check pattern of Region $(A)$ in the forth coming paper~\cite{EMS} using the fact 
that the localization factor of the Grover walk is the infinite energy flow of the given graph~\cite{HS}. 
We believe that a mathematical formulation of this phase pattern and also rigorous proof of each pattern are candidate of 
future's interesting problems. 
\begin{figure}[htbp]
  \begin{center}
              \includegraphics[clip, width=10cm]{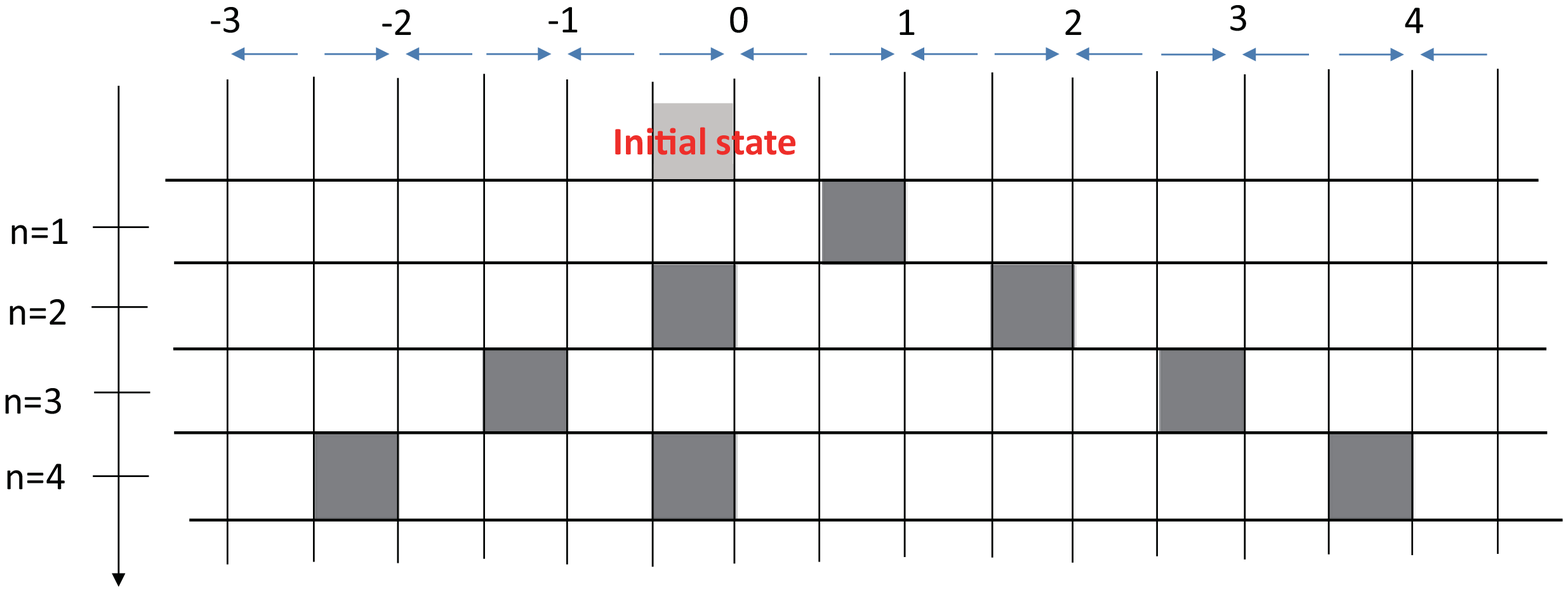}
              \caption{The phase pattern up to $n=4$ for $k=20$: Each cell corresponds to $(n,\; (j;N))$ with $n\in \mathbb{N}$, $j\in \mathbb{Z}$ and $N\in \{L,R\}$. 
	      If $\nu_n(j;N)=0$, that is, $\epsilon_j=1$ for $N=R$ and $\tau_j=1$ for $N=L$, then the corresponding cell color is black, otherwise the color is white. }
              \label{Fig.1}
  \end{center}
\end{figure}
\begin{figure}[htbp]
  \begin{center}
    \begin{tabular}{cc}

      \begin{minipage}{0.5\hsize}
        \begin{center}
          \includegraphics[clip, width=8cm]{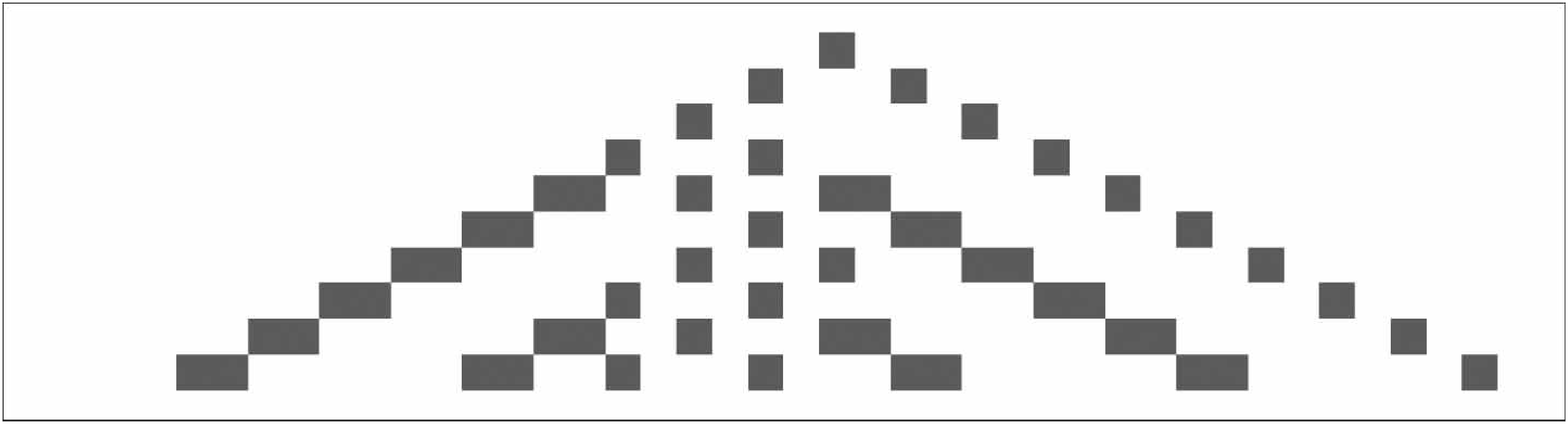}
          \hspace{1.6cm} (a)The phase pattern up to $n=10$
        \end{center}
      \end{minipage}
&
      \begin{minipage}{0.5\hsize}
        \begin{center}
          \includegraphics[clip, width=8cm]{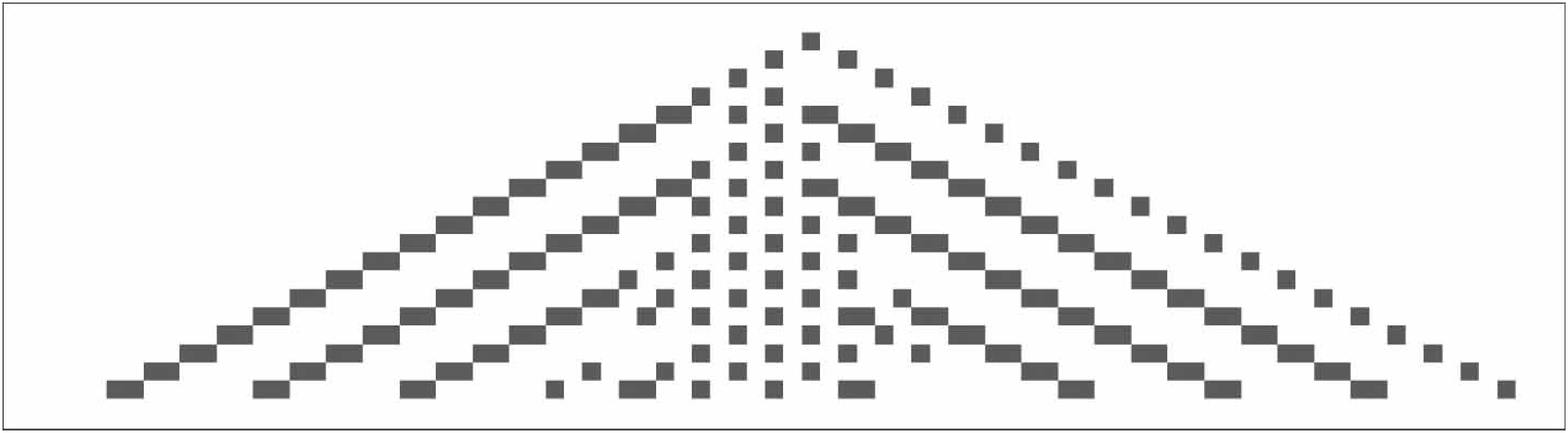}
          \hspace{1.6cm} (b)The phase pattern up to $n=20$
        \end{center}
      \end{minipage}
\\
\\
      \begin{minipage}{0.5\hsize}
        \begin{center}
          \includegraphics[clip, width=8cm]{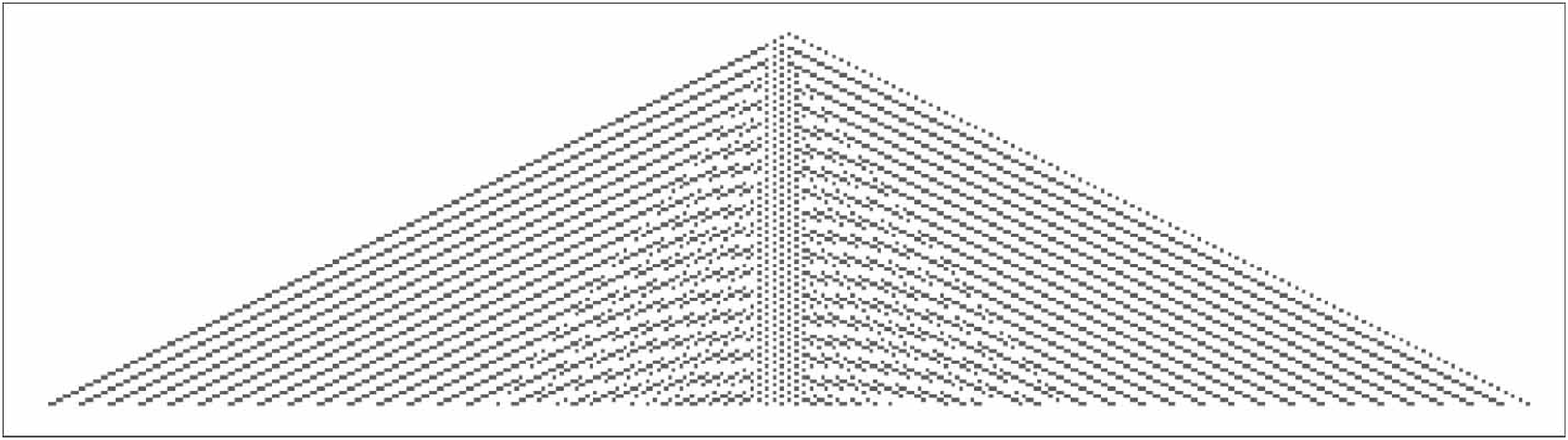}
          \hspace{1.6cm} (c)The phase pattern up to $n=100$
        \end{center}
      \end{minipage}

    \end{tabular}
    \caption{ Figures (a),(b) and (c) are the phase patterns up to $n=10$, $n=20$ and $n=100$, respectively for $k=20$ case}
    \label{Fig.2}
  \end{center}
\end{figure}
%
\begin{figure}[htbp]
  \begin{center}
    \begin{tabular}{cc}

      \begin{minipage}{0.5\hsize}
        \begin{center}
          \includegraphics[clip, width=8cm]{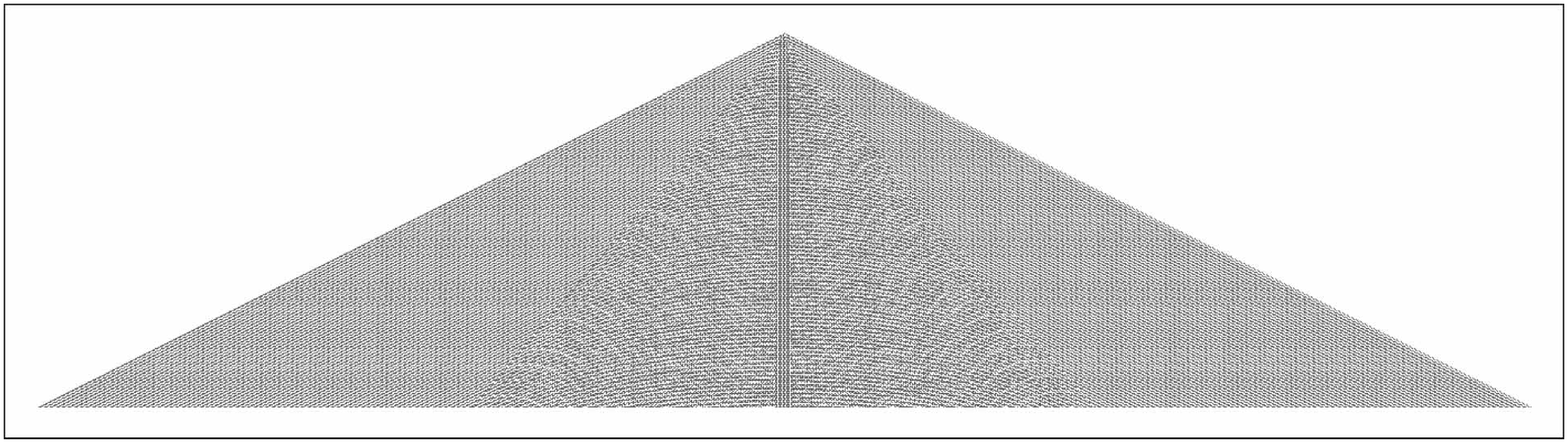}
          \hspace{1.6cm} The phase pattern for $k=20$
        \end{center}
      \end{minipage}
&
      \begin{minipage}{0.5\hsize}
        \begin{center}
          \includegraphics[clip, width=8cm]{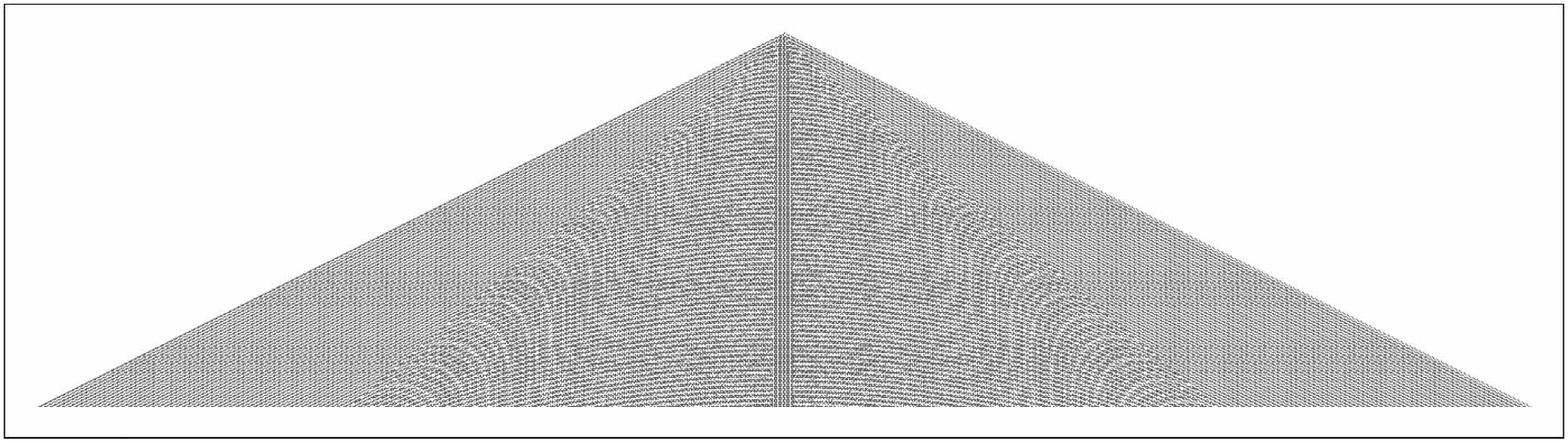}
          \hspace{1.6cm} The phase pattern for $k=10$
        \end{center}
      \end{minipage}
\\
\\
      \begin{minipage}{0.5\hsize}
        \begin{center}
          \includegraphics[clip, width=8cm]{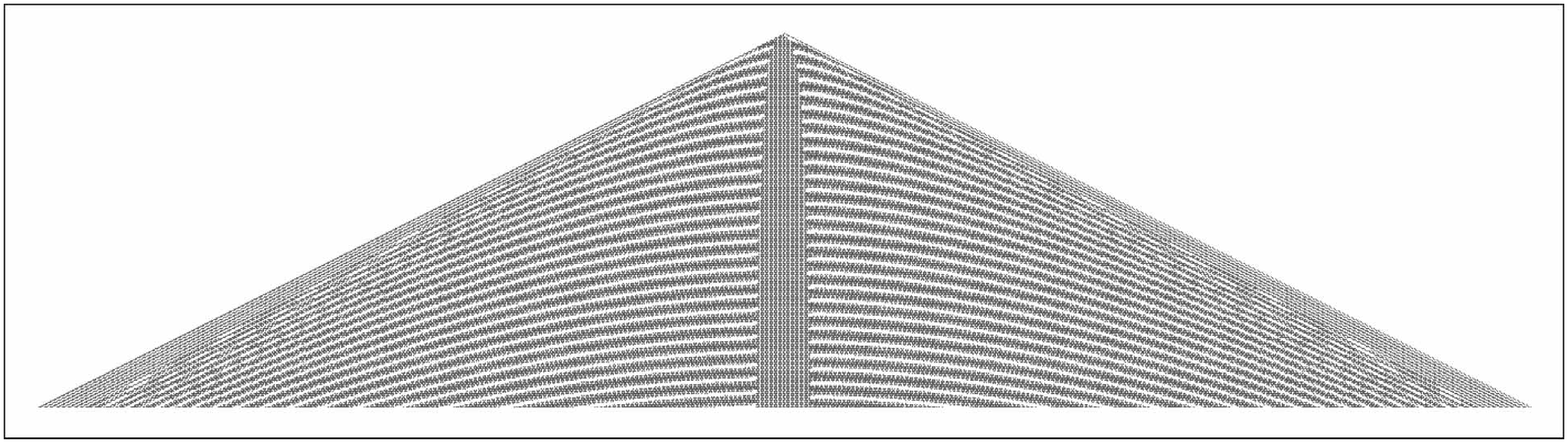}
          \hspace{1.6cm} The phase pattern for $k=3$
        \end{center}
      \end{minipage}

    \end{tabular}
    \caption{ Figures (a),(b) and (c) are the phase patterns Upton $n=500$, and whose degrees are $k=3$, $k=10$ and $k=20$, respectively. 
    The value $c$ which determines the boundary of the region $B$ and $C$ increases when $k$ approaches to a small value. 
    The value $c$ seems to be $2\sqrt{k-1}/k$ which is the diagonal part of the discriminant quantum walk's coin. }
    \label{Fig.3}
  \end{center}
\end{figure}
\section{Proof of main theorem} 
First, we prepare notations for the proof of the theorem. 
Secondly, we give the proof for $k$-regular tree case.  
After this consideration, we address to the proof for $k$-regular graph with $g(G)>2(n-1)$ by using the local tree structure. 
\subsection{Preliminary for the proof of main theorem}
For given $G=(V,E)$, 
let $A=A(G)$ be the set of symmetric arcs induced by the edge set $E$. 
\segawa{
A sequence of arcs $(a_1,\dots,a_{r})$ in $G$ satisfying $t(a_1)=o(a_2),\dots,t(a_{r-1})=o(a_{r})$ is called a $r$-length path or simply a path. 
The length of a shortest path from $u$ to $v$ is denoted by $\mathrm{dist}(u,v)$. 
If $t(a_r)=o(a_1)$, then this is called a $r$-length closed path. 
A back track of a closed path $(b_0,\dots,b_{s-1})$ is a subsequence $(b_j,b_{j+1})$ with $b_j=\bar{b}_{j+1}$, where $j\in\mathbb{Z}/s\mathbb{Z}$. 
If there are no back tracks in a closed path, the closed path is called a cycle. 
If there are no cycles in $G$ and $G$ is connected, then $G$ is called a tree. 
When $G$ is a tree, then the depth of the tree is the minimum hight, 
where the hight of the tree with a root $o\in V$ is defined by the largest distance from $o$. 
The $m$-level set of the tree with the root $o$ is the set of all the vertices whose distance from $o$ is $m$. 
The girth denoted by $g(G)$ is the length of a shortest cycle in $G$; 
when $G$ is a tree, we define $g(G)=\infty$. }

\segawa{
We define a distance of pair of arcs $(a,b)$ as follows: 
	\[ \mathrm{dist}(a,b):=\min\{ \mathrm{dist}(t(a),o(b)),\;\mathrm{dist}(t(a),t(b)),\;\mathrm{dist}(o(a),t(b)),\;\mathrm{dist}(o(a),o(b)) \}. \]
}
Every pair of arcs $(a,b)$ has at least one of the following positional relations (see also Fig.~\ref{positional}): \\
For $|a|\neq |b|$, 
	\begin{enumerate}
        \item\label{to} ${}^\exists j_1\in \mathbb{N}\cup\{0\}$ such that \segawa{$\mathrm{dist}(a,b)=j_1$ with } $t(a) \stackrel{j_1}{\sim} o(b)$; 
        \item\label{tt} ${}^\exists j_2\in \mathbb{N}\cup\{0\}$ such that \segawa{$\mathrm{dist}(a,b)=j_2$ with } $t(a) \stackrel{j_2}{\sim} t(b)$; 
        \item\label{ot} ${}^\exists j_3\in \mathbb{N}\cup\{0\}$ such that \segawa{$\mathrm{dist}(a,b)=j_3$ with } $o(a) \stackrel{j_3}{\sim} t(b)$;
        \item\label{oo} ${}^\exists j_4\in \mathbb{N}\cup\{0\}$ such that \segawa{$\mathrm{dist}(a,b)=j_4$ with } $o(a) \stackrel{j_4}{\sim} o(b)$;
        \end{enumerate}
and for $|a|=|b|$, (5) $a=\bar{b}$; (6) $a=b$. \\
Here for $u,v\in V$, $u\stackrel{j}{\sim}v$ denotes $\mathrm{dist}(u,v)=j$.
For cases (\ref{to})--(\ref{oo}), since ${}^TS(U)=JS(U)J$, it holds 
	\[ (S(U)^{\segawa{j_1+1}})_{b,a}>0,\;(JS(U)^{j_2+1})_{b,a}>0,\;({}^TS(U)^{j_3+1})_{b,a}>0,\;(J\;{}^TS(U)^{\segawa{j_4+1}})_{b,a}>0,   \]
\segawa{for cases (5) and (6), $(J)_{b,a}>0$, $(I)_{b,a}>0$, respectively. }
\segawa{We should remark that there is a possibility that $(S(U)^{j_1+1})_{b,a}>0$ and $(JS(U)^{j_2+1})_{b,a}>0$, simultaneously 
since there might be a cycle which accomplishes both positional relations (1) and (2). }
\begin{figure}[htbp]
  \begin{center}
              \includegraphics[clip, width=5cm]{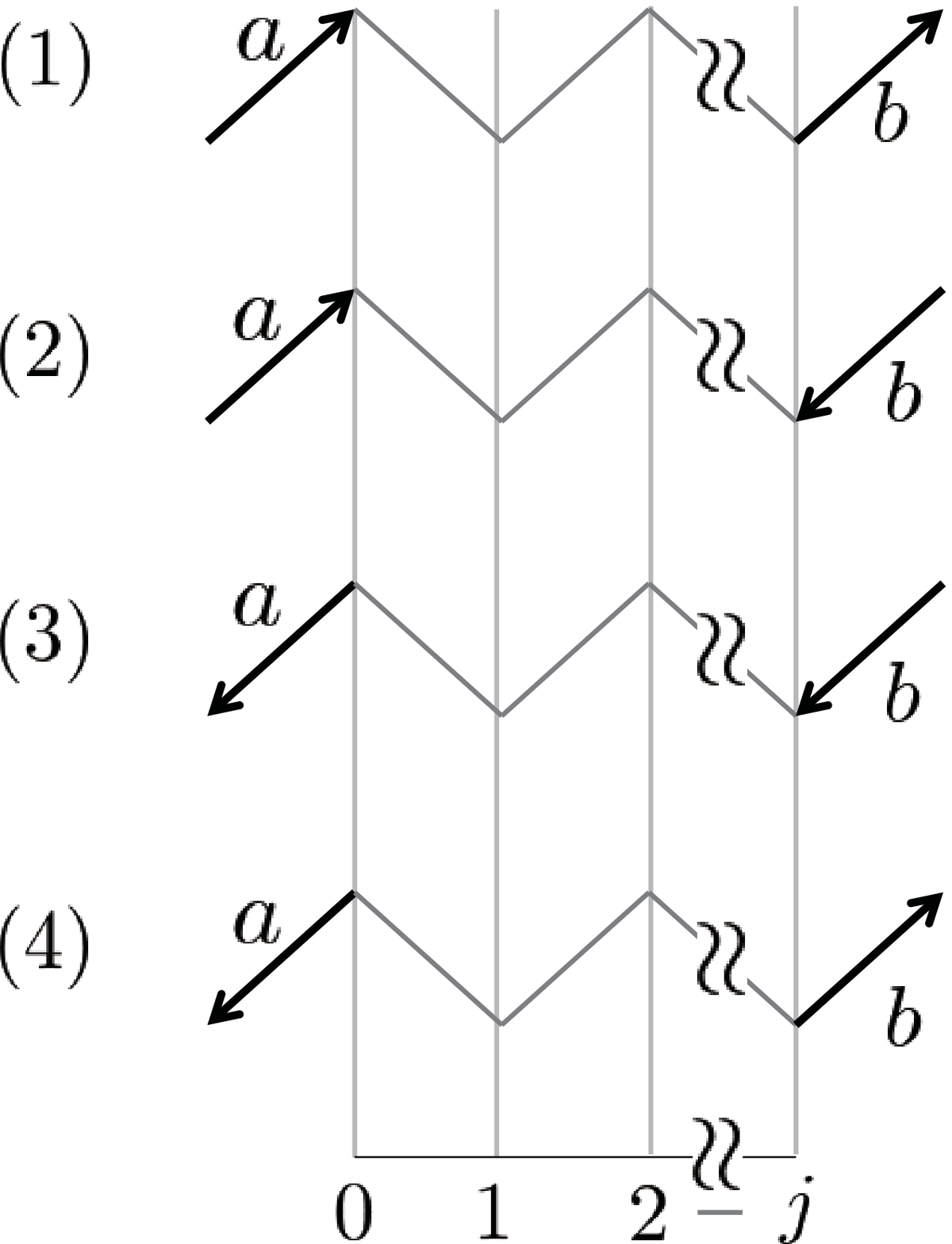}
              \caption{The positional relations of cases (1)--(4)}
  \label{positional}
  \end{center}
\end{figure}
\subsection{$k$-regular tree case}
If the graph is a tree, the positional relation of arbitrary pair of arcs $(a,b)$ is uniquely applicable to one of the cases (1)--\segawa{(6)}: 
if not, a cycle would appear in the tree. 
In the same reason, if the graph is a tree, it holds 
\begin{align*} 
(S(U)^{j_1})_{b,a}>0\Leftrightarrow (S(U)^{j_1})_{b,a}=1, & \; {}^\forall j\neq j_1,\; (S(U)^{j})_{b,a}=0 \\
(JS(U)^{j_2})_{b,a}>0\Leftrightarrow (JS(U)^{j_2})_{b,a}=1, & \; {}^\forall j\neq j_2,\; (JS(U)^{j})_{b,a}=0 \\
(J\;{}^TS(U)^{j_3})_{b,a}>0\Leftrightarrow (J\;{}^TS(U)^{j_3})_{b,a}=1, & \; {}^\forall j\neq j_3,\; (J\;{}^TS(U)^{j})_{b,a}=0 \\
({}^TS(U)^{j_4})_{b,a}>0\Leftrightarrow ({}^TS(U)^{j_4})_{b,a}=1. & \;{}^\forall j\neq j_4,\; ({}^TS(U)^{j})_{b,a}=0
\end{align*}
We summarize the above observation as follows. 
\begin{lemma}\label{lem1}
Assume that the graph is a tree. 
Let $\Xi$ be the following set of matrices 
	\[ \{ S(U)^j,\;JS(U)^{j+1},\;J\;{}^TS(U)^{j+1},\;{}^TS(U)^j \;|\; j\in \mathbb{N}\cup\{0\} \}. \]
Then we have 
 \segawa{
 \begin{align*}
 (a,b) \mathrm{\;is\; the\;} & \mathrm{positional\; relation\;} (1) \\ 
                            & \Leftrightarrow {}^\exists j\;s.t.,\; (S(U)^{j+1})_{b,a}=1, \;{}^\forall \xi\in \Xi\setminus \{S(U)^{j+1}\} ,\; (\xi)_{b,a}=0; \\
 (a,b) \mathrm{\;is\; the\;} & \mathrm{positional\; relation\;} (2) \\
                            & \Leftrightarrow {}^\exists j\;s.t.,\; (JS(U)^{j+1})_{b,a}=1, \;{}^\forall \xi \in \Xi\setminus \{JS(U)^{j+1}\},\; (\xi)_{b,a}=0; \\
 (a,b) \mathrm{\;is\; the\;} & \mathrm{positional\; relation\;} (3) \\
                            & \Leftrightarrow {}^\exists j\;s.t.,\; ({}^TS(U)^{j+1})_{b,a}=1,\;{}^\forall \xi \in \Xi\setminus \{{}^TS(U)^{j+1}\},\; (\xi)_{b,a}=0; \\
 (a,b) \mathrm{\;is\; the\;} & \mathrm{positional\; relation\;} (4)  \\
                            & \Leftrightarrow {}^\exists j\;s.t.,\; (J\;{}^TS(U)^{j+1})_{b,a}=1, \;{}^\forall \xi \in \Xi\setminus\{J{}^TS(U)^{j+1}\},\; (\xi)_{b,a}=0; \\
 (a,b) \mathrm{\;is\; the\;} & \mathrm{positional\; relation\;} (5)  \\
                            & \Leftrightarrow (J)_{b,a}=1, \;{}^\forall \xi \in \Xi\setminus\{J\},\; (\xi)_{b,a}=0; \\
 (a,b) \mathrm{\;is\; the\;} & \mathrm{positional\; relation\;} (6)  \\
                            & \Leftrightarrow (I)_{b,a}=1, \;{}^\forall \xi \in \Xi\setminus\{I\},\; (\xi)_{b,a}=0. 
 \end{align*}}
\end{lemma}
Let $\mathbb{T}$ be the $k$-regular tree (which is an infinite graph). 
\segawa{From now on we fix an arbitrary fixed arc $e$ as the initial arc and consider $S(U^n)\delta_e$, 
where $\delta_e$ is the delta function on $e$, that is, 
\[ \delta_e(f)=\begin{cases} 1 & \text{: $e=f$,}\\ 0 & \text{: $e\neq f$.} \end{cases} \]
For the $k$-regular tree, we formally describe $(\xi)_{b,a}:=\langle \delta_b, \xi \delta_a \rangle_{\ell^2(A(\mathbb{T}))}$
for $\xi\in \Xi$. }
We take an isometric deformation of $\mathbb{T}$ so that we can keep the symmetricity with respect to this fixed arc $e$: 
we ``unbend" the tree by moving all the \segawa{descendants} of $o(e)$ to the opposite side as follows: 
we decompose the vertices $V(\mathbb{T})$ into ($\dots,V_{-1},V_0,V_1,\dots$), where for $j\geq 1$, 
	\begin{align*}
        V_0 &= \{t(e)\}, \\
        V_{-j} &= \{ u \;|\; \mathrm{dist}(t(e),u)=j,\;\mathrm{dist}(o(e),u)=j-1 \},\\
        V_j &= \{ u \;|\; \mathrm{dist}(t(e),u)=j \} \setminus V_{-j} 
        \end{align*}
We also decompose the arc set $A(\mathbb{T})$ into ($\dots,A_{-1}^{R},A_0^{R},A_1^{R},\dots$), 
($\dots,A_{-1}^L,A_0^L,A_{1}^L,\dots$), where
	\begin{align*}
        A_j^R &= \{ a \;|\; t(a)\in V_j,\;o(a)\in V_{j-1} \}; \\
        A_j^L &= \{ a \;|\; t(a)\in V_j,\;o(a)\in V_{j+1} \}. 
        \end{align*}
\segawa{If both $a$ and $b$ belong to $A_j^{N}$, then $(a,e)$ and $(b,e)$ are the same positional relation. 
By Lemma~\ref{lem1}, such $\xi\in \Xi$ with $(\xi)_{a,e}=(\xi)_{b,e}=1$ is uniquely determined and represented by }
	\begin{align}\label{oneone1}
        \xi=
        	\begin{cases}  
                S(U)^j & \text{: $j\geq 0,\;N=R$;}\\
                JS(U)^{j+1} & \text{: $j\geq 0,\;N=L$;}\\
                {}^TS(U)^{\segawa{|j|}} & \text{: $j<0,\;N=R$;}\\
                J\;{}^TS(U)^{\segawa{|j+1|}} & \text{: $j< 0,\;N=L$.}
         	\end{cases} 
        \end{align}
Figure~\ref{one-one} depicts a simple chart of this one-to-one correspondence between $\mathbb{Z}\times\{L,R\}$ and $\Xi$. 
\begin{figure}[htbp]
  \begin{center}
              \includegraphics[clip, width=10cm]{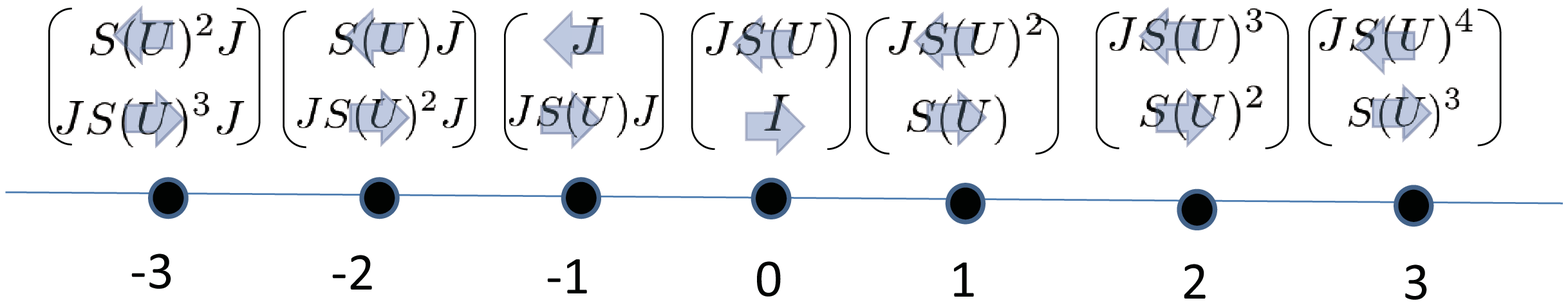}
              \caption{The one-to-one correspondence between $\mathbb{Z}\times\{L,R\}$ and $\Xi$: The right directed arrow on the position $j\in\mathbb{Z}$ 
              depicts $(j;R)$, that is, $A_j^R$,  and 
              left directed arrow on $j$ depicts $(j;L)$, that is, $A_j^L$. 
              The corresponding element of $\Xi$ to each arrow overlaps in this figure, e.g., 
              the arrow $(2;L)$ corresponds to $JS(U)^3$. }
              \label{one-one}
  \end{center}
\end{figure}
\segawa{
We obtained that 
if $a,b\in A_n^{N}$, then 
     $(S(U)^m)_{a,e}=(S(U)^m)_{b,e}$, 
    $(JS(U)^m)_{a,e}=(JS(U)^m)_{b,e}$, 
 $({}^TS(U)^m)_{a,e}=({}^TS(U)^m)_{b,e}$ and 
$(J\;{}^TS(U)^m)_{a,e}=(J\;{}^TS(U)^m)_{b,e}$ for every $m$. 
The following lemma provides that the amplitudes of the Grover walk itself $(U^m)_{a,e}$ and $(U^m)_{b,e}$ for $a,b\in A_j^{N}$
take also the same value. }
\begin{lemma}\label{lem2}
Assume $G$ is the $k$-regular tree. 
Let $\psi_n = U^n\delta_{e}$. 
Then 
	\[ \psi_n(a)=\psi_n(b) \] 
holds for any $a,b\in A_j^{N}$ with $j\in \mathbb{Z}$ and $N\in \{L,R\}$. 
\end{lemma}
\begin{proof}
This is easily completed by the induction with respect to the time iteration $n$. 
\end{proof}
Thus by (\ref{oneone1}) and Lemma~\ref{lem2}, $S(U^n)\delta_e$ 
is expressed by a linear combination of $\Xi\segawa{\delta_e:=\{\xi\delta_e \;|\; \xi\in \Xi\}}$ with the $\{0,1\}$-coefficient; 
\segawa{that is, 
	\begin{equation}\label{eq:from_e}
        S(U^n)\delta_e=\sum_{j=0}^n \left(\epsilon_j S(U)^j\delta_e +\tau_j JS(U)^j\delta_e\right)
        	+ \sum_{j=1}^{n-1} \left(\epsilon_{-j} {}^T(S(U)^j)\delta_e +\tau_{-j} J\;{}^T(S(U)^j)\delta_e\right),
        \end{equation} 
where $\epsilon_j,\tau_j\in\{0,1\}$. }
Moreover to obtain the $\{0,1\}$ coefficient of each term of $\Xi$, 
by Lemma~\ref{lem2}, we pick up the parity of only the representative of each $A_j^N$'s. 
To this end more efficiently, we introduce the following map: 
$\Psi: \ell^2(A(\mathbb{T}))\to \ell^2(\mathbb{Z}\times \{L,R\})$ such that 
	\begin{align*}
        (\Psi(\psi))(j;R) = \frac{1}{\sqrt{|A_j^R|}}\sum_{a\in A_j^R}\psi(a),\; (\Psi(\psi))(j;L) = \frac{1}{\sqrt{|A_j^L|}}\sum_{a\in A_j^L}\psi(a). 
        \end{align*}
We will regard $\Psi(\psi_n)(j;N)$ as the ``representative" of $\psi_n(a)$ for $a\in A_j^N$. 
A direct computation provides the following lemma. 
\begin{lemma}\label{lem3}
Let $\Psi$ and $\psi_n$ be the above. Then we have 
	\begin{align}
        (\Psi(\psi_n))(j;R)  &= \begin{cases}
        		\frac{2}{k}\sqrt{k-1}(\Psi(\psi_{n-1}))(j-1;R)-\left(\frac{2}{k}-1\right)(\Psi(\psi_{n-1}))(j-1;L) & \text{: $j\geq 0$,}\\ 
                        \frac{2}{k}\sqrt{k-1}(\Psi(\psi_{n-1}))(j-1;R)+\left(\frac{2}{k}-1\right)(\Psi(\psi_{n-1}))(j-1;L) & \text{: $j<0$,}
                        \end{cases} \\
        (\Psi(\psi_n)(j;L)  &= \begin{cases}
        		 \left(\frac{2}{k}-1\right)(\Psi(\psi_{n-1}))(j+1;R)+\frac{2}{k}\sqrt{k-1}(\Psi(\psi_{n-1}))(j+1;L) & \text{: $j\geq 0$,}\\
                        -\left(\frac{2}{k}-1\right)(\Psi(\psi_{n-1}))(j+1;R)+\frac{2}{k}\sqrt{k-1}(\Psi(\psi_{n-1}))(j+1;L) & \text{: $j< 0$.}
                        \end{cases}
        \end{align}
\end{lemma}
Putting $\phi_n(j)={}^T[(\Psi(\psi_n))(j;L)\;\;(\Psi(\psi_n))(j;R) ]$, we obtain the difference equation of the discriminant quantum walk; that is, 
	\begin{align*}
        \phi_0(j) &= \delta_0(j)|R\rangle, \\
        \phi_{n+1}(j) &= P(j+1)\phi_n(j+1)+Q(j-1)\phi_n(j-1),
	\end{align*}
where $P(j)$ and $Q(j)$ are given by Definition~\ref{DesQW}. 
\begin{remark}\label{rem:phase}
Let $\psi_n$ be the $n$-th iteration of the Grover walk on $k$-regular tree starting from $\delta_e$, and  
let $\phi_n$ be its discriminant quantum walk at time $n$. Then 
	\begin{align}
        \psi_n(a) &= 
        	\begin{cases}
                \frac{1}{\sqrt{|A^R_j|}}\langle R, \phi_n(j)\rangle & \text{: $a\in A^R_j$,}\\
                \frac{1}{\sqrt{|A^L_j|}}\langle L, \phi_n(j)\rangle & \text{: $a\in A^L_j$.}
                \end{cases}
        \end{align}
This implies the phase is invariant with respect to $\Psi$. 
\end{remark}
Therefore \segawa{by (\ref{oneone1}) and Remark \ref{rem:phase}, it holds in the $k$-regular tree case that
the coefficients in (\ref{eq:from_e}) are given by
	\begin{equation}\label{eq:epsilontau}
        \langle R,\phi_n(j)\rangle>0 \Leftrightarrow \epsilon_j=1,\;
        \langle L,\phi_n(j)\rangle>0 \Leftrightarrow \tau_{j-1}=1. 
	\end{equation}
Since $e\in A(\mathbb{T})$ can be chosen as an arbitrary arc due to the symmetricity of the $k$-regular tree, 
the statement of (\ref{eq:epsilontau}) also holds if we choose another initial arc $f\in A(\mathbb{T})\setminus\{e\}$ and consider $S(U^n)\delta_f$
in the same way; 
which implies
\[S(U^n)[\delta_e,\delta_f,\dots]=\left(\sum_{j=0}^n \left(\epsilon_j S(U)^j +\tau_j JS(U)^j\right)
        	+ \sum_{j=1}^{n-1} \left(\epsilon_{-j} {}^T(S(U)^j) +\tau_{-j} J\;{}^T(S(U)^j)\right)\right)[\delta_e,\delta_f,\dots]. \]
Since $[\delta_e,\delta_f,\dots]=I$ with some appropriate computational basis order, the statement of our main theorem is true for the $k$-regular tree. 
}
\subsection{$k$-regular graph with $g(G)>2(n-1)$ case}
Next, we consider a $k$-regular graph whose girth is greater than $2(n-1)$. 
For \segawa{an arbitrary} fixed $e\in A(G)$, we define the directed subgraph $\mathbb{T}_e^{(n)}=(V',A')\subset G$ as follows: 
putting $t(e):=o$, we define $A'$ and $V'$ by 
	\begin{align*}
                A'= A_{in} \cup \partial A, \;\; V'= t(A_{in}),
        \end{align*}
where 	
	\begin{align*} 
        A_{in} &= \{ a\in A \;|\; \mathrm{dist}(o, t(a))\leq n-1,\; \mathrm{dist}(o,o(a))\leq n-2  \} \\
         \partial A & = \{ a\in A \;|\; \mathrm{dist}(o, t(a))=n,\; \mathrm{dist}(o,o(a))=n-1 \};
        \end{align*}
in this setting, we omit the terminal vertices of all $a\in \partial A$. 
\begin{lemma}\label{tree}
Let $G$ be a $k$-regular and $g(G)>2(n-1)$. 
Then $\mathbb{T}_e^{(n)}$ is a depth-$n$ directed tree whose root is $t(e)$. 
\end{lemma}
\begin{proof}
If there is a cycle in $\mathbb{T}_e^{(n)}$, then the cycle does not pass any arcs in $\partial A$ 
since all the terminal vertices $\partial A$ are vanished. 
The length of a largest cycle in this graph is less than the diameter of $G$, that is $2(n-1)+1$, 
which should run through each level set. 
To accomplish this cycle, we need the arc connecting two vertices in $(n-1)$-level set, but 
such an arc belongs to non available arcs in $\partial A$ by the definition. 
Then the largest length of the cycle is at \segawa{most} $2(n-1)$. 
By the assumption $g(G)>2(n-1)$, then the contradiction occurs. 
Thus there are no cycles. 
\end{proof}
Thus for the $k$-regular graph $G$ whose girth is greater than $2(n-1)$, when we look around from an arbitrary vertex, 
we can regard it as a local $k$-regular tree within $(n-1)$-distance from the vertex. 
For a graph $H$, the time evolution of the Grover walk on the graph $H$ is denoted by $U_H$. 
By Lemma~\ref{tree}, the following statement is immediately obtained. 
\begin{lemma}
Let $G$ be $k$-regular with $g(G)>2(n-1)$ and also let $\mathbb{T}_e^{(n)}$ be the above directed subtree with depth $n$. 
Then for every $e\in A(G)$, and for every $f\in A'$, 
	\[ (U_G^j)_{f,e}=(U_{\mathbb{T}_e^{(n)}}^j)_{f,e}\;\;(j=0,1,\dots,n) \]
\end{lemma}
\begin{proof}
If $f\in A_{in}$, the statement is trivial. We show it for $f\in \partial A$ case. 
To reach $f\in \partial A$ from the initial arc $e$, it takes at least $n$ iterations. 
If there are two shortest paths from $e$ to $f$; $(e,e_1,\dots,e_{n-1},f)$ and $(e,e_1',\dots,e_{n-1}',f)$, 
then since $o(e_1)=o(e_1')=t(e)$ and $t(e_{n-1})=t(e_{n-1}')=o(f)$, the $2(n-1)$-length cycle appears. 
This is contradiction to the assumption $g(G)>2(n-1)$.  
\end{proof}
Therefore when the Grover walk on the $k$-regular graph $G$; whose girth is greater than $2(n-1)$, starts from $\delta_e$, 
we can convert it to the Grover walk on the $k$-regular tree starting from $e$ whose terminus is the root of the tree 
as long as the time iteration is less than $n$. 
\segawa{Since $e\in A(G)$ can be chosen as an arbitrary arc, 
the same statement also holds if we choose another initial arc $e'\in A(G)\setminus\{e\}$ and consider $S(U^n)\delta_{e'}$. }
It is completed the proof. 

\section{Spectrum of $S(U^n)$}
For given $k$-regular graph $G=(V,E)$, we consider the spectrum of $S(U^n)$. 
\segawa{We show the spectrum of $S(U^n)$ is inherited from the spectrum of the underlying graph. 
Although the setting of \cite{MOS}, whose setting was for an extended Szegedy walk, 
is different from our setting, our fundamental analytical method can follow \cite{MOS} 
with some modifications to be able to apply to our setting. 
The operation ``$S$" will remove regularities of the unitary operator $U^n$.
This method enables us to find when $S(U^n)$ loses the diagonalizable property. }

We introduce $K:\ell^2(A)\to \ell^2(V)$ as the following boundary operator corresponding to an incidence matrix with respect to the terminal vertices: 
	\[ K=\sum_{a\in A} |t(a)\rangle\langle a | \]
where $|u\rangle$ and $|a\rangle$ denote the standard basis of $\ell^2(V)$ and $\ell^2(V)$ corresponding to $u\in V$, $a\in A$, respectively.
This boundary operator $K$ plays the important role for our spectral analysis due to the following properties:
\begin{enumerate}
\item $S(U)=J(K^*\;K-I)$; 
\item $KK^*=k \bs{1}_{|V|}$; 
\item $KJK^*=M$, where $M$ is the adjacency matrix of $G$. 
\end{enumerate}  
We put $L: \ell^2(V)\times \ell^2(V)\to \ell^2(A)$ such that $L=[ K^*, \; JK^* ]$. 
Using the above properties of $K^*$ and $J^2=\bs{1}_{|A|}$, we obtain
	\begin{equation} 
        S(U)L=L\tilde{M} 
        \end{equation}
Here $\tilde{M}:\ell^2(V)\times \ell^2(V)\to \ell^2(V)\times \ell^2(V)$ is 
	\[ \tilde{M}=\begin{bmatrix} 0 & -1 \\ k-1 & M \end{bmatrix}. \]
Remarking that ${}^TS(U)=JS(U)J$, we can easily obtain that 
	\begin{align}
        S(U)^jL &= L\tilde{M}^j; \notag \\ 
        JS(U)^jL &= L\tilde{\sigma}_X\tilde{M}^j; \notag \\
        {}^TS(U)^jL &= L\tilde{\sigma}_X\tilde{M}^j\tilde{\sigma}_X; \notag \\
        J\;{}^TS(U)^jL &= L\tilde{M}^j\tilde{\sigma}_X. \label{map}
        \end{align}
Here $\tilde{\sigma}_X=\sigma_X\red{\otimes I_{|V|}}$ with \[ \sigma_X=\begin{bmatrix}0 & 1 \\ 1 & 0\end{bmatrix}. \]
Let $F(x;y)$ be a linear combination of $\{x^j,yx^j,yx^jy,x^jy \;|\; j\in\mathbb{Z}\}$. 
Then by (\ref{map}), 
	\begin{equation}\label{linear}
        F(S(U);J)L = LF(\tilde{M};\tilde{\sigma}_X).
        \end{equation}
Let $\mathcal{L}\subset \ell^2(A)$ be $L(\ell^2(V)\times \ell^2(V))$, that is, $\mathcal{L}=\{K^*f+JK^*g \;|\; f,g\in \ell^2(V)\}$ 
By (\ref{linear}), we have the following equivalent deformation with respect to the eigenequation of $F(S(U);J)$ restricted to $\mathcal{L}$ ;
	\begin{align}\label{lee}
        F(S(U);J)|_{\mathcal{L}}\psi = \lambda \psi,\;\; \psi\neq 0
        	& \Leftrightarrow L(\lambda-F(\tilde{M};\tilde{\sigma}_X))\phi = \red{0},\;\;\phi\notin \ker L
        \end{align}
Thus we need to clarify $\ker L$. Indeed this is expressed as follows. 
\begin{lemma}
	Let $L$ and $\tilde{M}$ be the above. Then we have 
        \[ \ker(L)=\ker(1-\tilde{M}^2)=\ker\left( k+\begin{bmatrix} 0 & M \\ M & 0 \end{bmatrix} \right). \]
\end{lemma}
\begin{proof}
First we show $\ker(L)\subset \ker(1-\tilde{M}^2)$. 
Let ${}^T[f.g]\in \ker(L)$. Then $K^*g+JK^*g=0$ holds. By multiplying $K$ and $KJ$, then we have 
	\begin{equation}\label{ker} kf+Mg=0,\;\;Mf+kg=0 \end{equation}
Using this equation, a simple computation leads 
	\[ (1-\tilde{M}^2){}^T[f,g]=0 \]
which implies ${}^T[f,g]\in \ker(1-\tilde{M}^2)$, that is, $\ker(L)\subset \ker(1-\tilde{M}^2)$. 
Next we show $\ker(L)\supset \ker(1-\tilde{M}^2)$. 
We assume ${}^T[f,g]\in \ker(1-\tilde{M}^2)$. 
Then 
	\[ kg+Mg=0,\;Mf+g=0. \]
Remarking $KK^*=k$ and $KJK^*=M$, we can provides the following equivalent expression of the above equations: 
	\[ K(K^*f+JK^*g)=0,\;KJ(K^*f+JK^*g)=0 \]
which are also equivalent to $K^*f+JK^*g\in \mathcal{L}^\perp$. Thus $K^*f+JK^*g$ should be $0$ since $K^*f+JK^*g\in \mathcal{L}$. 
It is completed the proof. 
\end{proof}
Since the adjacency matrix $M$ is a regular matrix, then $M$ can be decomposed into
	\[ M=\sum_{\mu\in \sigma(M)}\mu \Pi_\mu, \]
where $\Pi_\mu$ is the orthogonal projection onto the eigenspace of $\mu$. 
Therefore the equivalent deformations of the eigenequation (\ref{lee}) can be continued to 
	\begin{multline}\label{eq:bunkai}
        (1-\tilde{M}^2)(\lambda-F(\tilde{M};\tilde{\sigma}_X) ) \phi=0,\;\;\phi\notin \ker (1-\tilde{M}^2) \\
        	\Leftrightarrow 
                \left(\sum_{\mu\in \sigma(M)}(1-K_\mu^2)(\lambda-F(K_\mu;\sigma_X))\otimes \Pi_\mu \right) \phi=0,\;\;\phi\notin \ker (1-\tilde{M}^2) 
        \end{multline}
Here $K_\mu$ is the $2$-dimensional matrix defined by 
	\[ K_\mu = \begin{bmatrix} 0 & -1 \\ k-1 & \mu \end{bmatrix}. \]
We are interested in so-called ``non-trivial" zeros of 
	\[ \mathrm{det}(\lambda-S(U^n))=0 \]
not living in the real line, which is a graph analogue of the non-trivial poles of the Riemann zeta function; 
that is the zeros of the Ihara zeta~\cite{Bass1992,Ihara1966,KS2000} motivated by the quantum walks. 
Remark that all the eigenvalues for the eigenspace $\mathcal{L}^\perp$ are included in $\mathbb{R}$ because 
$\mathcal{L}^\perp=\ker(K)\cap \ker(KJ)=\left\{\ker(K)\cap \ker(1-J)\right\}\cap \left\{\ker(K)\cap \ker(1+J)\right\}$. 
Therefore to see the non-trivial eigenvalues,  we can concentrate on the eigenequation:
\begin{theorem}\label{thm:eigenorbit}
Let $\lambda$ be the non-trivial eigenvalue of $S(U^n)$. 
Then the value $\lambda$ satisfies 
	\begin{equation}\label{eigenorbit}
        \det(\lambda-F_n(K_\mu;\sigma_X))=0 
        \end{equation}
Here $F_n$ is determined by (\ref{mainThmEq1}). 
\end{theorem}
\begin{proof}
\red{
Let us consider the decomposition (\ref{eq:bunkai}). 
We consider the solution of the eigenequation 
	\[ \det\left(\sum_{\mu\in \sigma(M)}(1-K_\mu^2)(\lambda-F(K_\mu;\sigma_X))\otimes \Pi_\mu \right)=0. \]
The above equation is equivalent to 
	\[ \prod_{\mu\in \sigma(M)}\det(1-K_\mu^2)(\lambda-F(K_\mu;\sigma_X))=0. \]
First we easily notice that $|1-K_\mu^2|=0$ iff $\mu=\pm k$. 
Then if $\mu\neq \pm k$, the solution of $\det(\lambda-F(K_\mu;\sigma_X))=0$ with respect to $\lambda$ is an eigenvalue of $F(S(U);J)$. 
Then we showed that the statement is true at least for $\mu\neq \pm k$. 
So secondly we consider $\mu=k$ case. 
Since 
	\[\ker(1-K_\mu^2)=\ker\begin{bmatrix} k & \mu \\ \mu & k \end{bmatrix},  \]
we have 
	\[ \ker(1-K_k^2)(\lambda-F(K_k;\sigma_X))=\left\{ \begin{bmatrix}\alpha \\ \beta \end{bmatrix} \;|\; (a+c-\lambda)\alpha+(b+d-\lambda)\beta=0 \right\} \]
when we put \[F(K_k;\sigma_X)=\begin{bmatrix} a & b \\ c & d \end{bmatrix}. \]
Moreover it is easy to compute that 
	\[ K_k^j=\frac{1}{k-2}\begin{bmatrix} (k-1)-(k-1)^j & 1-(k-1)^j \\ (k-1)^{j+1}-(k-1) & -1+(k-1)^{j+1} \end{bmatrix}, \]
which implies $a+c=b+d$. Then if $a+c-\lambda\neq 0$; that is, $\lambda\neq (k-1)^n(k-2)$, then 
	\[ \ker(1-K_k^2)(\lambda-F(K_k;\sigma_X))=\left\{ \begin{bmatrix}\alpha \\ \beta \end{bmatrix} \;|\; \alpha+\beta=0 \right\}. \]
However in that case, we can check that $\phi\in \ker(1-{\tilde{M}}^2)$. Thus $\lambda$ must be $(k-1)^n(k-2)$, which is a real value and not a non-trivial eigenvalue. 
The case for $\mu=-k$ can be done in the same way as $\mu=k$ case.
Therefore $\mu=\pm k$ cases can be excluded. 
}
\end{proof}
Theorem~\ref{mainThm} shows the concrete expression of $F_n$ for each $n$. 
For example, for $n=3$ case, (\ref{eigenorbit}) is reduced to 
	\[ \det(\lambda-(K_\mu^3+\sigma K_\mu \sigma))=0 \]
We put the solution $\lambda(\mu):=x(\mu)+\im y(\mu)$ with $x(\mu),y(\mu)\in \mathbb{R}$.
Since $F_n(\mu):=F_n(K_\mu;\sigma_X)$ is a $2$-dimensional matrix, the solution of (\ref{eigenorbit}) with respect to $\lambda$ is that of 
the following characteristic polynomial: 
	\[ \lambda^2-\mathrm{tr}(F_n(\mu))\lambda+\det(F_n(\mu))=0. \]
Define $D_n(\mu):=\mathrm{tr}(F_n(\mu))^2-4\det(F_n(\mu))$. 
The real part; $x(\mu)$, of this solution $\lambda(\mu)$ lies on 
	\begin{multline}\label{eq:Re}
        R_n:=\{ (\mu,x)\in [-k,k]\times \mathbb{R} \;|\; x^2-\mathrm{tr}(F_n(\mu))x+\det(F_n(\mu))=0\red{,\; D_n(\mu)\geq 0} \} \\
        \cup \{ (\mu,x)\in [-k,k]\times \mathbb{R} \;|\; \mathrm{tr}(F_n(\mu))-2x=0,\;D_n(\mu)\leq 0 \}
        \end{multline}
On the other hand, the imaginary part; $y(\mu)$, lies on  the following algebraic equation which draws a kind of hyperelliptic curve for $n\geq 2$:
	\begin{align}\label{eq:Im}
        I_n:=\{ (\mu,y)\in [-k,k]\times \mathbb{R} \;|\; 4y^2+D_n(\mu)=0 \}.
        \end{align}
Therefore the non-trivial eigenvalues of $S(U^n)$ are inherited from all the eigenvalues of the underlying graph satisfying $D_n(\mu)\leq 0$. 
In Appendix, explicit expressions for $R_n(\mu)$ and $I_n(\mu)$ for $n=1,\dots,6$ are described. 
The parity of $D_n(\cdot)$ determines the range of the spectrum of $S(U^n)$. 
On the other hand, the following theorem shows the affect of zero's of $D_n(\cdot)$ on the matrix property of $S(U^n)$. 
\begin{theorem}
A zero of $D_n(\cdot)$; \red{$\mu_*$}, belongs to $\sigma(M)$ \red{and $F_n(\mu_*)\neq cI_2$ with some nonzero constant $c$} if and only if $S(U^n)$ is not diagonalizable. 
\end{theorem}
\begin{proof}
This is obtained by the fact 
that the geometric multiplicity of the eigenvalue $\lambda(\mu)$ of $K_\mu$ is strictly less than the algebraic multiplicity for such $\mu$. 
\red{ If $\ker(\lambda-F_n(\mu))\subset \ker(\lambda-F_n(\mu))^2$, then $D_n(\mu)=0$ since the algebraic multiplicity must exist. 
On the other hand, if $D_n(\mu)=0$ then since $F_n(\mu)$ is a $2$-dimensional matrix, the only case that $\ker(\lambda-F_n(\mu))= \ker(\lambda-F_n(\mu))^2$ is the case $F_n(\mu)=cI_2$. }
\end{proof}

Thus if the spectrum of the underlying graph has the branch point of $R_n$, then the induced $S(U^n)$ becomes non-diagonalizable.  

Finally we draw the support of non-trivial eigenvalues with the parameter $\mu\in [-k,k]\supset \sigma(M)$ in the complex plane, that is, 
$\{(x(\mu),y(\mu)) \;|\; \mu\in[-k,k],\;D_{n}(\mu)\leq 0\}$ in Fig.~\ref{Fig.4} 
and \segawa{$I_n$ and $R_n$ to show how $\Re(\sigma(S(U)^n))$ and $\Im(\sigma(S(U)^n))$ are 
inherited from $\sigma(M)$ and the value $\mu_*$ which makes $S(U^n)$ be non-diagonalizable in Fig.~\ref{Fig.7} for $n=1,2,\dots,6$. }
\begin{figure}[htbp]
  \begin{center}
    \begin{tabular}{ccc}

      \begin{minipage}{0.33\hsize}
        \begin{center}
          \includegraphics[clip, width=3.3cm]{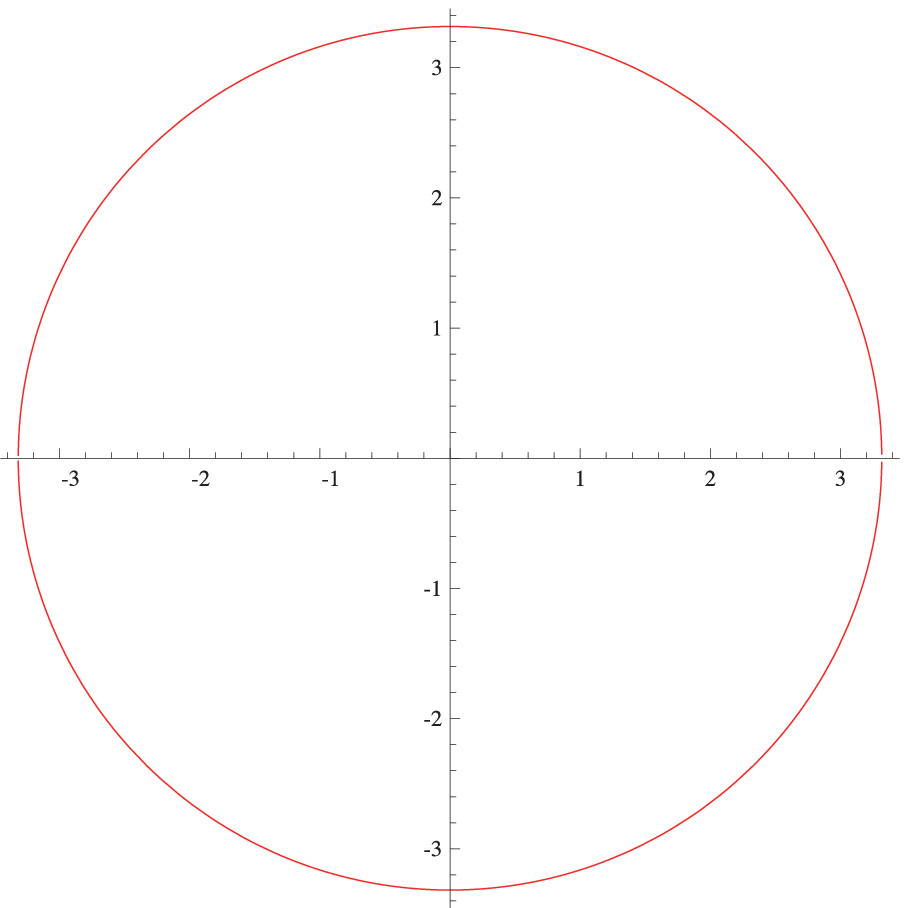}
          \hspace{1.6cm} (1)The support of eigenvalue of $S(U)$
        \end{center}
      \end{minipage}
&
      \begin{minipage}{0.33\hsize}
        \begin{center}
          \includegraphics[clip, width=3.3cm]{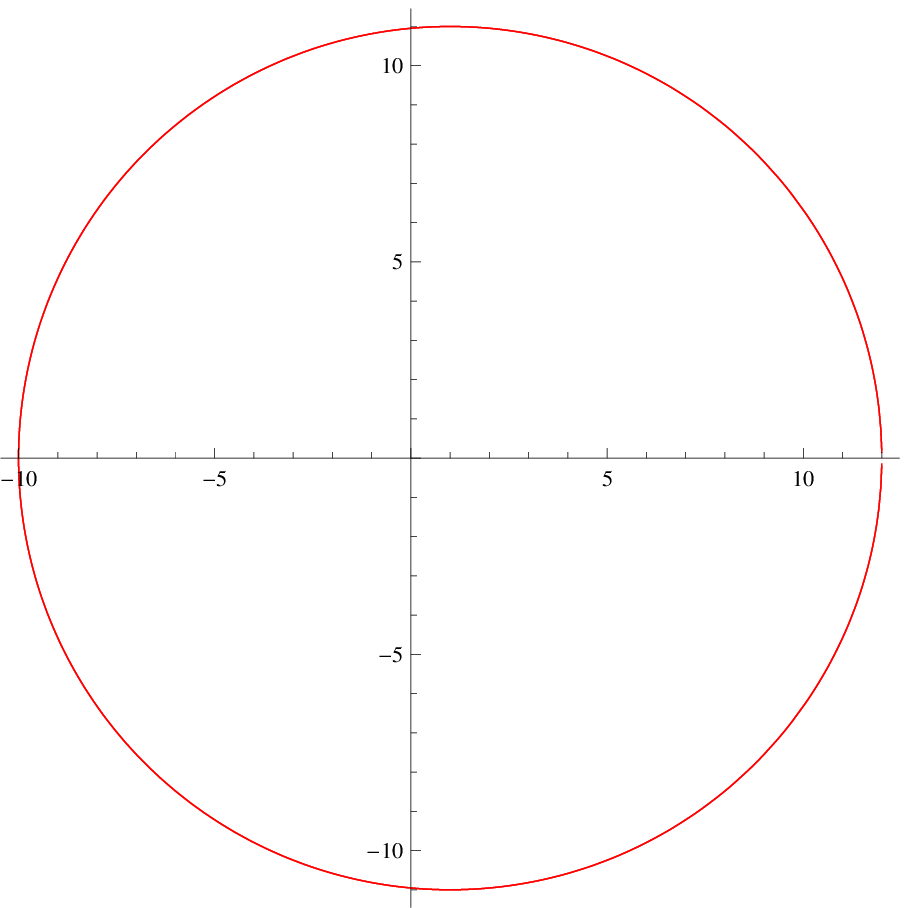}
          \hspace{1.6cm} (2)The support of eigenvalue of $S(U^2)$
        \end{center}
      \end{minipage}
&
      \begin{minipage}{0.33\hsize}
        \begin{center}
          \includegraphics[clip, width=3.3cm]{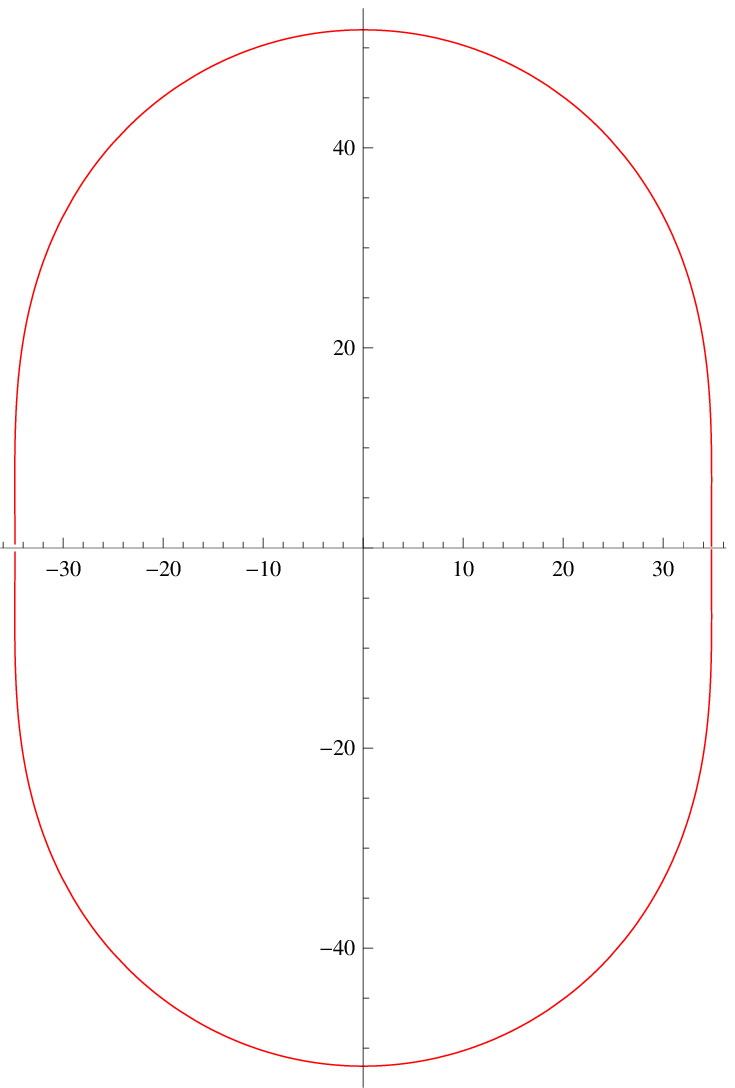}
          \hspace{1.6cm} (3)The support of eigenvalue of $S(U^3)$
        \end{center}
      \end{minipage}
\\
\\
      \begin{minipage}{0.33\hsize}
        \begin{center}
          \includegraphics[clip, width=3.3cm]{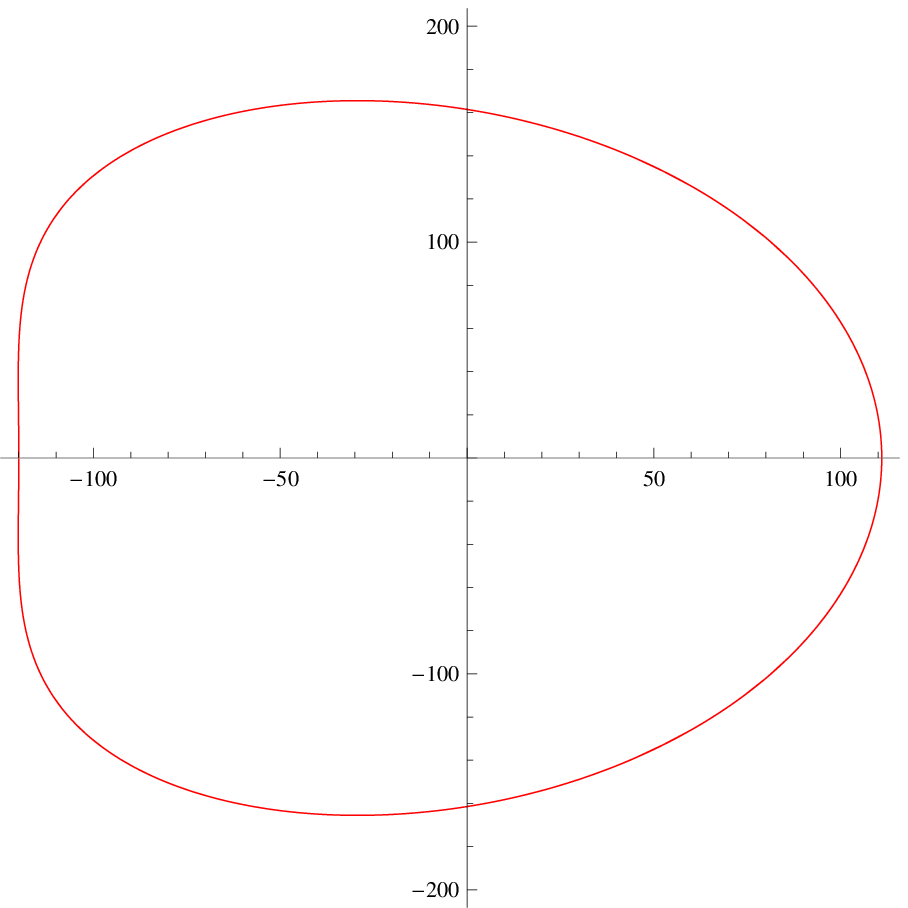}
          \hspace{1.6cm} (4)The support of eigenvalue of $S(U^4)$
        \end{center}
      \end{minipage}
&
      \begin{minipage}{0.33\hsize}
        \begin{center}
          \includegraphics[clip, width=3.3cm]{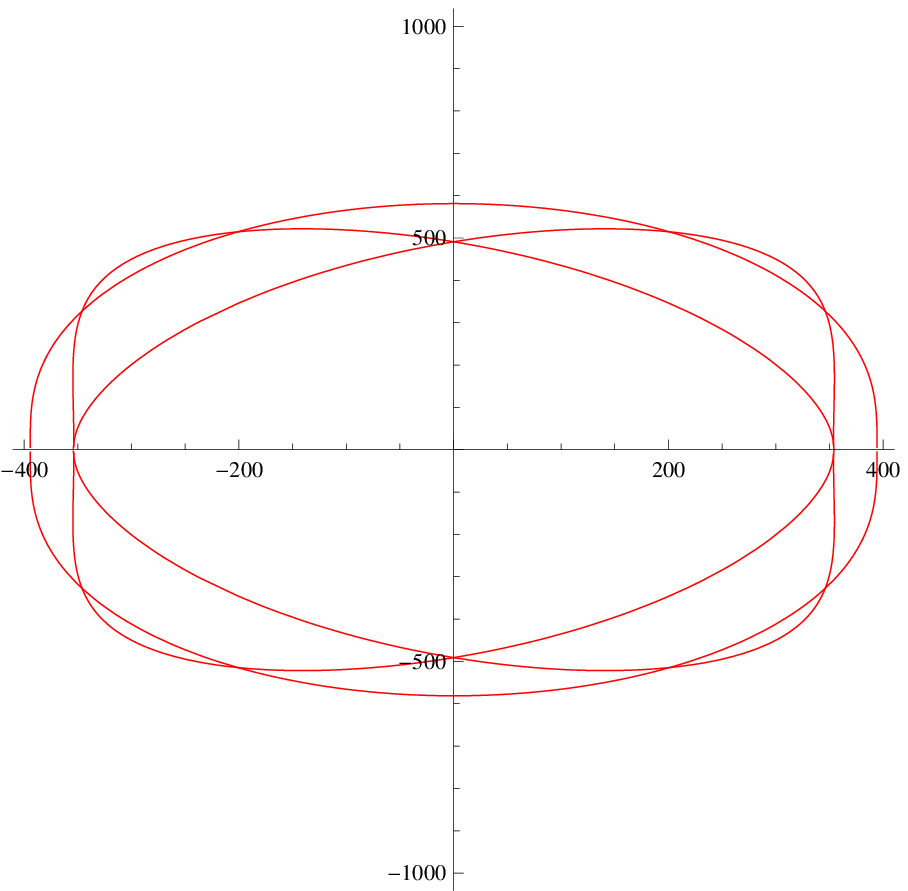}
          \hspace{1.6cm} (5)The support of eigenvalue of $S(U^5)$
        \end{center}
      \end{minipage}
&
      \begin{minipage}{0.33\hsize}
        \begin{center}
          \includegraphics[clip, width=3.3cm]{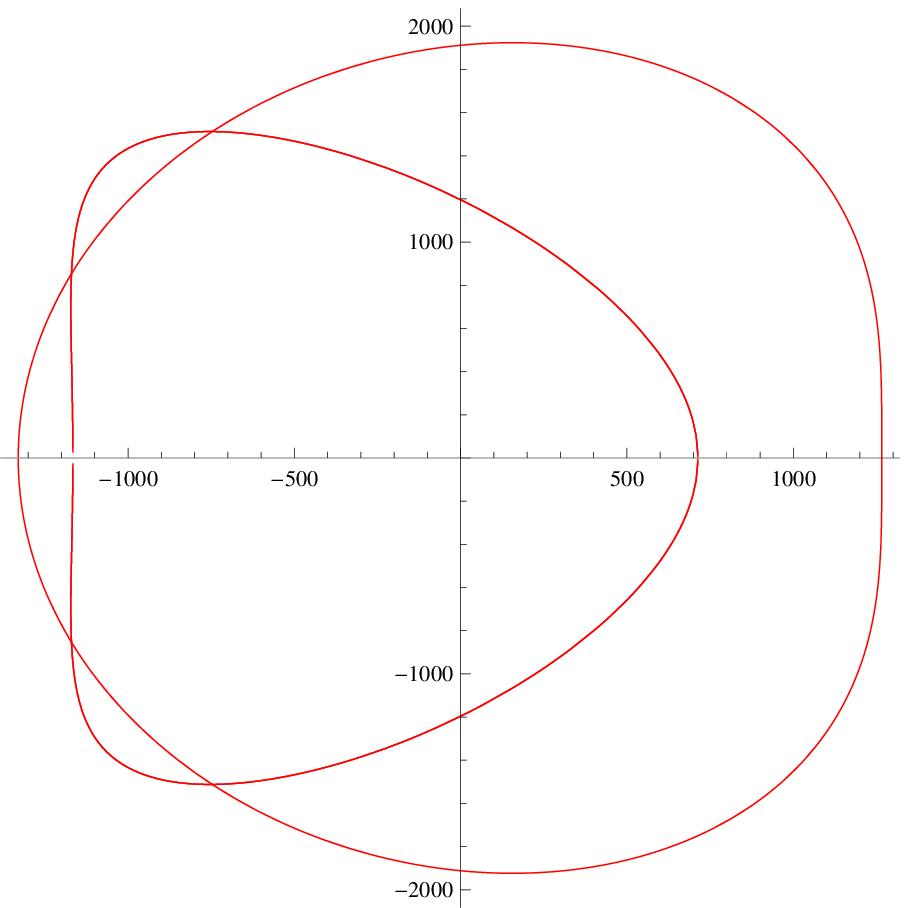}
          \hspace{1.6cm} (6)The support of eigenvalue of $S(U^6)$
        \end{center}
      \end{minipage}
    \end{tabular}
    \caption{ Figures (1)--(6) are the support of the non-trivial eigenvalues in $\mathbb{C}$ of $S(U),S(U^2),\dots,S(U^6)$, respectively for $k=12$ and $g(G)>2(n-1)$. 
    The horizontal and vertical lines are real and imaginary lines, respectively.
    The eigenvalues of $S(U^n)$ must lie on this support for each $n$. }
    \label{Fig.4}
  \end{center}
\end{figure}
%
%
\begin{figure}[htbp]
  \begin{center}
    \begin{tabular}{ccc}

      \begin{minipage}{0.33\hsize}
        \begin{center}
          \includegraphics[clip, width=5cm]{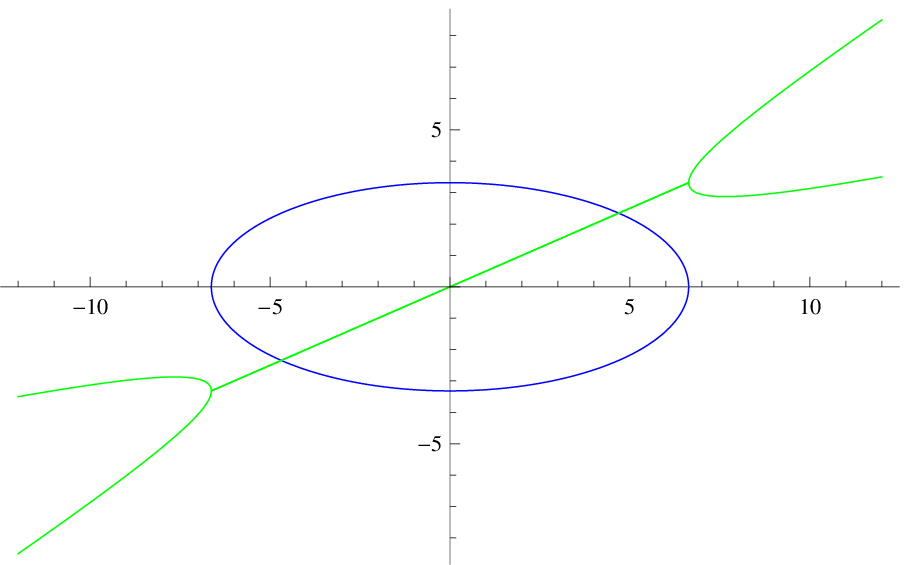}
          \hspace{1.6cm} {\scriptsize $I_1$ and $R_1$ for $k=12$} 
        \end{center}
      \end{minipage}
&
      \begin{minipage}{0.33\hsize}
        \begin{center}
          \includegraphics[clip, width=5cm]{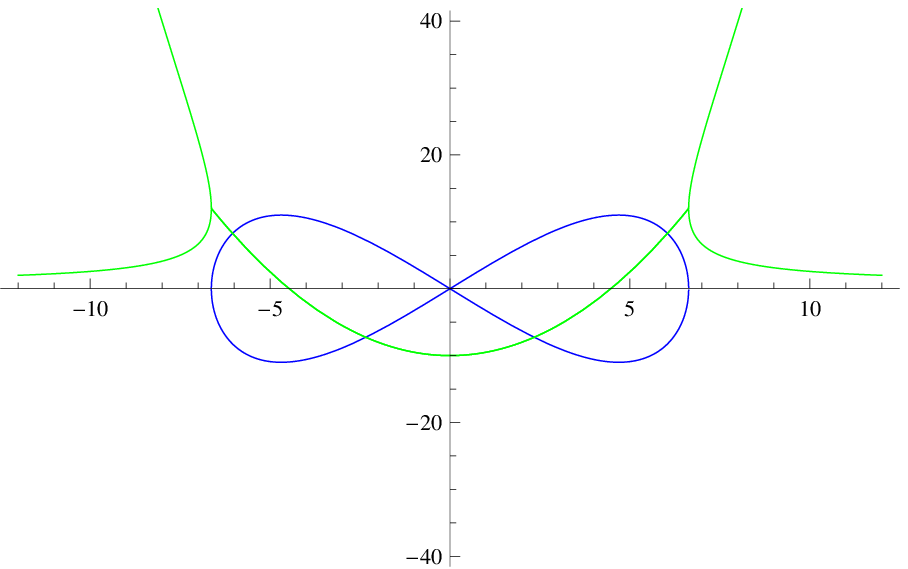}
          \hspace{1.6cm} {\scriptsize $I_2$ and $R_2$ for $k=12$}
        \end{center}
      \end{minipage}
&
      \begin{minipage}{0.33\hsize}
        \begin{center}
          \includegraphics[clip, width=5cm]{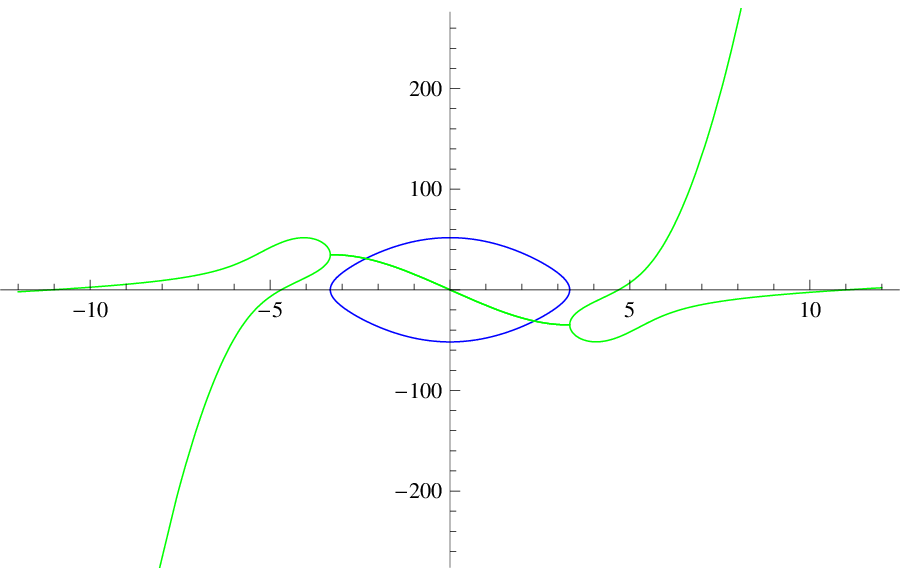}
          \hspace{1.6cm}  {\scriptsize $I_3$ and $R_3$ for $k=12$}
        \end{center}
      \end{minipage}
\\
\\
      \begin{minipage}{0.33\hsize}
        \begin{center}
          \includegraphics[clip, width=5cm]{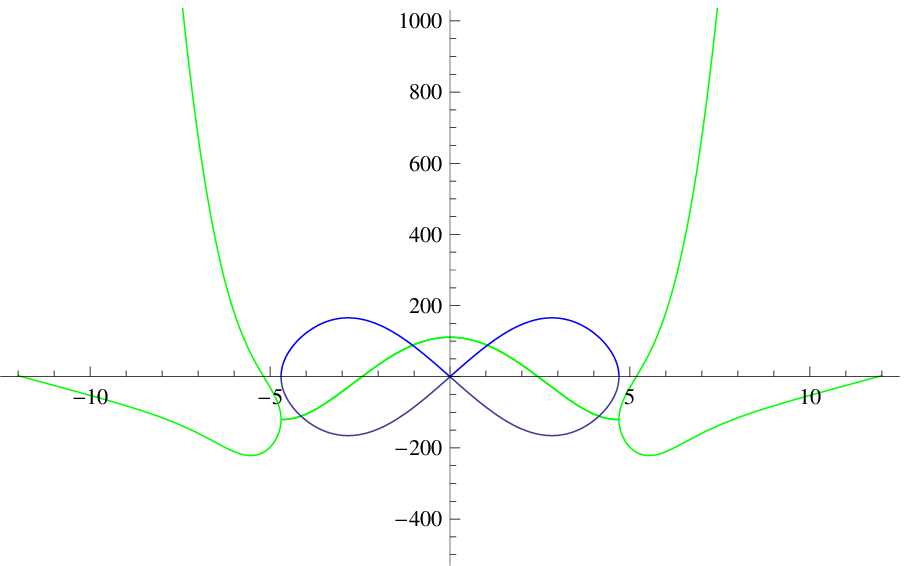}
          \hspace{1.6cm} {\scriptsize $I_4$ and $R_4$ for $k=12$}
        \end{center}
      \end{minipage}
&
      \begin{minipage}{0.33\hsize}
        \begin{center}
          \includegraphics[clip, width=5cm]{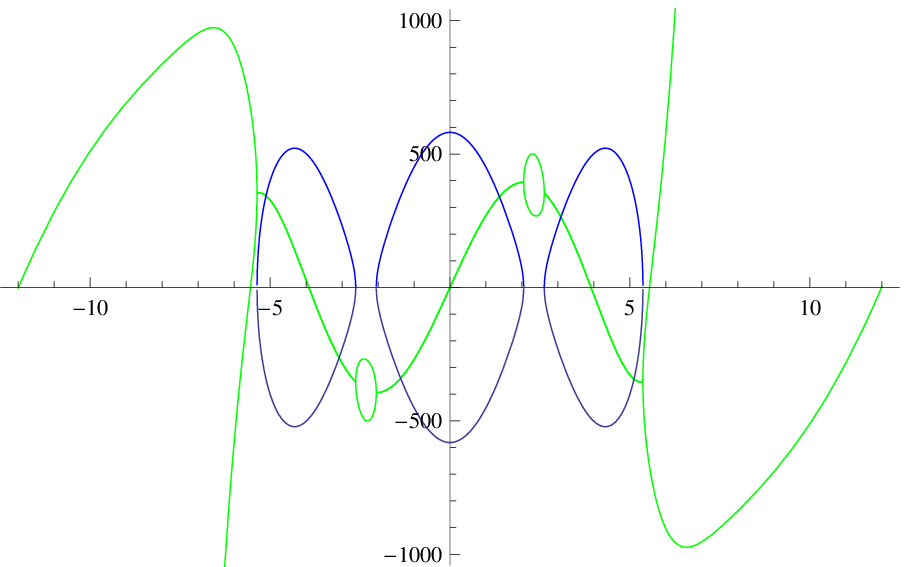}
          \hspace{1.6cm}  {\scriptsize $I_5$ and $R_5$ for $k=12$}
        \end{center}
      \end{minipage}
&
      \begin{minipage}{0.33\hsize}
        \begin{center}
          \includegraphics[clip, width=5cm]{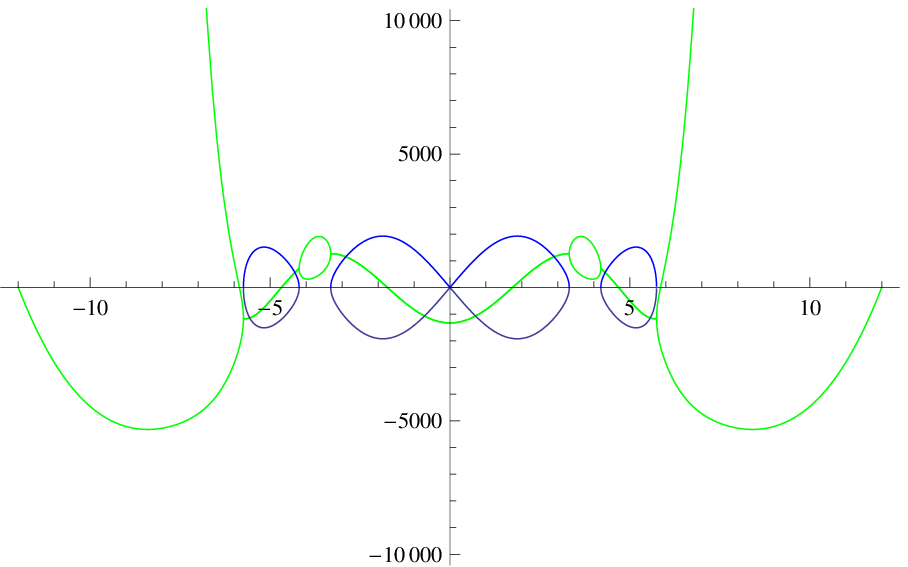}
          \hspace{1.6cm} {\scriptsize $I_6$ and $R_6$ for $k=12$}
        \end{center}
      \end{minipage}
    \end{tabular}
    \caption{ The orbits of real and imaginary parts of eigenvalue of $S(U^n)$ for $k=12$ and $g(G)>2(n-1)$ ($n=1,\dots,6$)
    obtained by (\ref{eq:Re}) and (\ref{eq:Im}), respectively: 
    The horizontal line is the parameter $\mu\in[-k,k]\supset \sigma(M)$, 
    and the vertical line is the real and imaginary parts of eigenvalue of $S(U^n)$; $R_n$ and $I_n$, respectively.  
    The green curve depicts $R_n$ and the blue curve depicts $I_n$. 
    Thus $R_n$ and $I_n$ show how the spectrum of the graph is mapped to the real and imaginary parts of that of induced $S(U^n)$. 
    The spectrum of the graph has the branch points in $R_n$ if and only if $S(U^n)$ is non-diagonalizable.}
    \label{Fig.7}
  \end{center}
\end{figure}

\section{Summary}
We obtained the structure theorem of the $n$-th power of the Grover walk operator on $G$ (Theorem~\ref{mainThm})
and its support of the non-trivial spectrum for the girth $g(G)>2(n-1)$ (Theorem~\ref{thm:eigenorbit}). 
The non-trivial spectrum is not living on the real line, which is a graph analogue of the non-trivial poles of the Riemann zeta function. 
We showed that this problem is converted to solving the phase pattern of the one-dimensional quantum walk in Definition~\ref{DesQW} 
which is only determined by the regularity of $G$. 
The curious phase pattern can be seen in Fig.~3 and the support of the spectrum can be seen in Fig.~4. 
Solving rigorously this phase pattern is one of the interesting future's problems.

\appendix
\def\thesection{Appendix \Alph{section}}
\renewcommand{\theequation}{A.\arabic{equation}}
\setcounter{equation}{0}

\section{}
The curve $R_n$ for $D_n(\mu)\leq 0$ is the set of the zero's of the following polynomial $Q_n$ with respect to $\mu$ and $x$:
\[ Q_1=\mu-2x \]
\[ Q_2=(4-2 k)+\mu^2-2x \]
\[ Q_3=(4-3 k) \mu+\mu^3-2x \]
\[ Q_4=\left(6-6 k+2 k^2\right)+(5-4 k) \mu^2+\mu^4-2x \]
\[ Q_5=(6 - 11 k + 5 k^2) \mu + (6 - 5 k) \mu^3 + \mu^5-2x \]
\[ Q_6=\left(4-6 k+6 k^2-2 k^3\right)+\left(10-20 k+9 k^2\right) \mu^2+(7-6 k) \mu^4+\mu^6-2x \]
The curve $I_n$ is the set of the zero's of the following polynomial $P_n$ with respect to $\mu$ and $y$ which draws a hyperelliptic curve for $n\geq 2$, 
where $D_n(\mu)$ is obtained by $P_n-4y^2$. 
\[ P_1=(4-4 k)+\mu^2+4y^2 \]
\[ P_2=(4-4 k) \mu^2+\mu^4+4y^2 \]
\[ P_3=-8 \left(-2+4 k-3 k^2+k^3\right)+\left(16-24 k+13 k^2\right) \mu^2+(4-6 k) \mu^4+\mu^6+4y^2 \]
\[ P_4=12 \left(3-7 k+6 k^2-2 k^3\right) \mu^2+\left(25-44 k+24 k^2\right) \mu^4+(6-8 k) \mu^6+\mu^8+4y^2 \]
\begin{multline*} 
P_5=
(16 - 48 k + 76 k^2 - 68 k^3 + 36 k^4 - 8 k^5) + (48 - 152 k + 205 k^2 - 146 k^3 + 41 k^4) \mu^2 \\
+ (52 - 156 k + 174 k^2 - 66 k^3) \mu^4 + (28 - 66 k + 39 k^2) \mu^6 \\
+ (8 - 10 k) \mu^8 + \mu^{10}+4y^2
\end{multline*}
\begin{multline*}
 P_6=
 -12 ((-1 + k)^3 (3 - 7 k + 5 k^2)) \mu^2 + (100 - 428 k + 704 k^2 - 524 k^3 + 149 k^4) \mu^4 \\
 + (100 - 324 k + 358 k^2 - 136 k^3) \mu^6 + (45 - 100 k + 58 k^2) \mu^8 \\
 + (10 - 12 k) \mu^{10} + \mu^{12}+4y^2 
\end{multline*}
Remark that for $k\geq 5$, the forms of $\{Q_n\}$ and $\{P_n\}$ need classifications with respect to the value $k$; 
see (\ref{SU3})--(\ref{SU6}); the above forms are in the case of $k\geq 12$. 
The forms of $Q_n$ and $P_n$ for general $n$ can be obtained by running the discriminant quantum walk on $\mathbb{Z}$ until $n$-step, 
and taking the phase measurement; e.g., for $k=20$, the phase pattern until $n=100$ can be referred in Fig.~\ref{Fig.2}.

\begin{small}
\bibliographystyle{jplain}

\end{small}

\end{document}